\documentclass{asjcauth}

\usepackage{moreverb}

\usepackage{graphicx,soul,framed,epsfig,mathptmx,times,amsmath,amssymb,color,
flushend,gensymb,subfigure,fancyhdr,subfigure,lipsum}
\usepackage{cite}
\allowdisplaybreaks
\def\volumeyear{2019}
\def\issuemonth{May}
\def\volumenumber{21}
\def\issuenumber{3}
\def\startpage{1064}
\def\endpage{1076}
\def\DOI{asjc.1802}

 \DeclareUnicodeCharacter{221E}{\ensuremath{\infty}}

\newtheorem{proof}{Proof}

\newtheorem{remark}{Remark}
\newtheorem{lemma}{Lemma}
\newtheorem{assumption}{Assumption}
\newtheorem{theorem}{Theorem }
\usepackage[utf8]{inputenc}

\usepackage{framed}
\definecolor{shadecolor}{rgb}{1,0.8,0.3}

\begin{document}

\runninghead{Erkan~Kayacan~and~Thor~I.~Fossen: Feedback Linearization Control for Systems with Mismatched Uncertainties via Disturbance Observers}


\title{FEEDBACK LINEARIZATION CONTROL FOR SYSTEMS WITH MISMATCHED UNCERTAINTIES VIA DISTURBANCE OBSERVERS}

\author{Erkan~Kayacan~and~Thor~I.~Fossen}

\address{E. Kayacan is with Coordinated Science Lab, University of Illinois at Urbana-Champaign, Urbana, Illinois 61801, USA. e-mail: {\tt\small erkank@illinois.edu } \\
T. I. Fossen is with Department of Engineering Cybernetics, Center for Autonomous Marine Operations and Systems, Norwegian University of Science and Technology, Trondheim, Norway. e-mail: {\tt\small thor.fossen@ntnu.no }}


\begin{abstract}
This paper focuses on a novel feedback linearization control (FLC) law based on a self-learning disturbance observer (SLDO) to counteract mismatched uncertainties. The FLC based on BNDO (FLC-BNDO) demonstrates robust control performance only against mismatched time-invariant uncertainties while the FLC based on SLDO (FLC-SLDO) demonstrates robust control performance against mismatched time-invariant and -varying uncertainties, and both of them maintain the nominal control performance in the absence of mismatched uncertainties. In the estimation scheme for the SLDO, the BNDO is used to provide a conventional estimation law, which is used as being the learning error for the type-2 neuro-fuzzy system (T2NFS), and T2NFS learns mismatched uncertainties. Thus, the T2NFS takes the overall control of the estimation signal entirely in a very short time and gives unbiased estimation results for the disturbance. A novel learning algorithm established on sliding mode control theory is derived for an interval type-2 fuzzy logic system. The stability of the overall system is proven for a second-order nonlinear system with mismatched uncertainties. The simulation results show that the FLC-SLDO demonstrates better control performance than the traditional FLC, FLC with an integral action (FLC-I) and FLC-BNDO. 
\end{abstract}

\keywords{Disturbance observer, feedback linearization control, mismatched uncertainty, uncertain systems.}

\maketitle


%
%

\section{INTRODUCTION}

Uncertainties, e.g., unmodelled dynamics, parameter variations, and disturbances, are inevitable in real-time applications \cite{7302059}. Since control performance is violated by these uncertainties, disturbance attenuation has a major influence on control system design \cite{ASJC1649,KAYACAN20141,ASJC1489,7934317}. Therefore, control algorithms based on disturbance observers have been developed to counteract uncertainties in systems \cite{ASJC1347, Gao2017,KAYACAN2017}. In this approach, the main idea is to merge all uncertainties into one term and to add the estimated value of the disturbance by a disturbance observer into a control law that is derived for a disturbance-free model. This control structure was successfully applied in many applications, such as reusable launch vehicles \cite{hall2006sliding}, direct drive motors \cite{Komada1991}, robot manipulators \cite{liu2000disturbance, Yonghwan1999}, ball mill grinding circuits \cite{Chen20091205}, pulse width modulated inverter \cite{Yokoyama1994}, inverted pendulum \cite{Chen2010}, nonuniform wind turbine tower \cite{He2015}, uninterruptible power supplies \cite{Mattavelli2005} and dynamic positioning system of ships  \cite{Du2016}. Moreover, many controllers, such as adaptive fuzzy-neural control \cite{Leu1999} and sliding mode control (SMC) \cite{Chen2010, Yang2013}, have been designed based on disturbance observers. 

Uncertainties can be broadly categorized into two groups: matched and mismatched uncertainties \cite{KAYACAN201578}. Disturbance and control input appear in the same channel in the former case while they appear in different channels in the later case \cite{ginoya2015state, Yang2017}. In control of systems with matched uncertainties, there exist many successful applications. The reason is that since the control input is in the same channel with the disturbance, only estimated disturbance value is added into the traditional control law. However, traditional control methods based on a disturbance observer fail to drive states of systems with mismatched uncertainties to the desired equilibrium point so that there is a necessity to derive novel control laws \cite{Ginoya2014}. The first solution is to add an integral action to control laws to remove steady-state error \cite{Chen2016}. This results in robust control performance against mismatched time-invariant uncertainties but may violate the nominal performance in the absence of the mismatched. The other solution is to derive novel control laws. For example, a novel SMC method has been proposed in \cite{Yang2013} by proposing a new sliding surface for a second-order nonlinear systems with mismatched uncertainties. The simulation results show that the proposed control method provides robust control performance against mismatched uncertainties and also maintains the nominal performance in the absence of mismatched uncertainties.

Another novel SMC method for nth order nonlinear systems with mismatched uncertainties in \cite{Ginoya2014}, both the memoryless and memory-based integral sliding surfaces and integral sliding-mode controllers for continuous-time linear systems with mismatched uncertainties in \cite{7498627}, a back-stepping control method for nonlinear systems with mismatched uncertainties in \cite{SUN20141027}, and $H_{\infty}$ controller based on a nonlinear DO in \cite{Wei2010} have been developed recently. Moreover, a comprehensive study of disturbance observer-based control methods has been given in \cite{7265050}.

In disturbance observer based control approach, uncertainties are represented as a parameter in system models, and thus online parameter estimation methods are required. There are nonlinear state and parameter estimation methods, such as particle filtering method and extended Kalman filter \cite{erkancenmpc,erkan2016acc,erkandenmpc,erkandinmpc}. However, it has been reported that the former one becomes practically unsolvable in real-time for a large number of variables while the latter one is not a good candidate in case of having strong non-linearities \cite{Haseltine2005, Rawlings2006}. Neuro-fuzzy structures with SMC theory-based learning algorithms have been employed for control and identification purposes in feedback-error learning (FEL) scheme \cite{Efe2000,erkan2013t1nfc,erdal2015t2fnn,6250790,6871316}. In this scheme, a conventional controller works in parallel with a neuro-fuzzy controller in which the output of the conventional controller is used as the learning error to train the neuro-fuzzy controller \cite{Gomi1993}. After a very short time, the neuro-fuzzy controller takes the overall control signal while the conventional controller converges to zero. Moreover, SMC theory-based learning algorithms ensure robustness and faster convergence rate than the traditional learning methods due to the computational simplicity. Motivated by these facts, an FEL estimation scheme is proposed in this study due to the computational simplicity and learning capability.

The main contributions of this paper beyond the existing state-of-the-art are as follows. First, a novel feedback linearization control (FLC) law based on a self-learning DO (SLDO) is developed against mismatched invariant and varying uncertainties, and the stability of the overall system is proven by taking the SLDO dynamics into account. Second, the SLDO is designed in FEL scheme in which the conventional estimation law and type-2 neuro-fuzzy structure (T2NFS) work in parallel while the former one is used as the learning error for the T2NFS. Thus, the FEL algorithm is employed to develop an observer for the first time. Third, novel SMC theory-based learning rules are proposed for the tuning of the T2NFS. In addition to the stability of the training algorithm, the stability of the SLDO is proven with Lyapunov stability analysis. To the best knowledge of the authors, this is also the first-time such a stability analysis is ever proven. Fourth, type-2 fuzzy logic systems are examined under noisy conditions, and computationally efficient disturbance observers are developed. 


The paper consists of five sections: The problem formulation is given in Section \ref{sec_prob_form}. The novel FLC law based on the SLDO is formulated, and the stability of the overall system is proven in Section \ref{sec_flcsldo}. The simulation results are presented in Section \ref{sec_simulation}. Finally, a brief conclusion is given in Section \ref{sec_conc}.


\section{Problem Formulation}\label{sec_prob_form}

A second-order nonlinear system with mismatched uncertainties is written in the following form:
\begin{eqnarray}\label{eq_nonlinearsystem}
\dot{x}_{1} &=& x_{2} + d \nonumber \\
\dot{x}_{2} &=& a(\textbf{x}) + b(\textbf{x}) u \nonumber \\
y &=& x_{1} \end{eqnarray}
where $x_{1}$ and $x_{2}$ denote the states, $u$ denotes the control input, $d$ denotes the disturbance, $y$ denotes the output, $a(\textbf{x})$ and $b(\textbf{x})$ denote the nonlinear system dynamics.

\subsection{Traditional FLC}

The traditional FLC law is written as follows:
\begin{equation} \label{eq_tflc_controllaw}
u =- b^{-1}(\textbf{x}) \Big( a(\textbf{x}) + k_{1} x_{1} + k_{2} x_{2} \Big)
\end{equation}
where $k_{1}, k_{2}$ denote the coefficients of the controller and they are positive, i.e., $k_{1}, k_{2}>0$.

If the control law in \eqref{eq_tflc_controllaw} is applied to the system with mismatched uncertainties in \eqref{eq_nonlinearsystem}, the closed-loop dynamics are obtained as follows:
\begin{equation}\label{eq_tflc_error}
\ddot{x}_{1} + k_{2} \dot{x}_{1} + k_{1} x_{1} = k_{2} d + \dot{d}
\end{equation}

\begin{remark}\label{remark_tflc} Equation \eqref{eq_tflc_error} shows that the state cannot be driven the desired equilibrium point under the control law proposed in \eqref{eq_tflc_controllaw}. This implies why the FLC is sensitive to mismatched uncertainties.
\end{remark}

\subsection{FLC with Integral Action}
An integral action is added to the control law in order to make the system robust against mismatched uncertainties. The FLC with an integral action (FLC-I) is formulated as follows:
\begin{equation}\label{eq_flci_controllaw}
u =- b^{-1}(\textbf{x}) \Big( a(\textbf{x}) + k_{1} x_{1} + k_{2} x_{2} + k_{i} \int^t_0 x_{1} \,dt\Big)  
\end{equation}
where $k_{i}$ denotes the coefficient for the integral action and it is positive, i.e., $k_{i}>0$. 

If the control law in \eqref{eq_flci_controllaw} is applied to the system with mismatched uncertainties in \eqref{eq_nonlinearsystem}, the closed-loop dynamics are obtained as follows:
\begin{equation}\label{eq_flci_error}
\dddot{x}_{1} + k_{2} \ddot{x}_{1} + k_{1} \dot{x}_{1} + k_{i} \dot{x}_{1} = k_{2} \dot{d} + \ddot{d}
\end{equation}
This implies that if the disturbance has a steady-state value, i.e., $\dot{d}=\ddot{d}=0$, the state can be driven to the desired equilibrium point.

\begin{remark}\label{remark_flci} Equation \eqref{eq_flci_error} shows that if disturbance has a steady-state value, i.e., $\dot{d}=\ddot{d}=0$, the FLC-I is robust to mismatched time-invariant uncertainties and thus the steady-state error is eliminated. However, adding an integral action can result in a worse performance than the nominal performance obtained by the FLC in the absence of mismatched uncertainties. 
\end{remark}


\subsection{FLC for Mismatched Time-Invariant Uncertainties}\label{sec_flcdo}

\subsubsection{Basic Nonlinear Disturbance Observer}

The second-order nonlinear system with mismatched uncertainties in \eqref{eq_nonlinearsystem} can be re-written in the following form:
\begin{eqnarray}\label{eq_nonlinearsystem_do}
\dot{\textbf{x}} &=& \textbf{g}_{1} (\textbf{x}) + \textbf{g}_{2} (\textbf{x}) u + \textbf{z} d
 \nonumber \\
y &=& x_{1} \end{eqnarray}
where $\textbf{x}=[x_{1}, x_{2}]^{T}$ denotes the state vector, $u$ denotes the control input, $d$ denotes the disturbance, $y$ denotes the output, $\textbf{g}_{1}(\textbf{x})=[x_{2}, a(\textbf{x})]^{T}$ and $\textbf{g}_{2}(\textbf{x}) =[0,b(\textbf{x})]^{T}$ denote the nonlinear system dynamics, and $\textbf{z}=[1,0]^{T}$ denotes the disturbance coefficient vector.

The BNDO has been proposed in \cite{chen2003nonlinear,Chen2004,Yang2011} for uncertainty problem to estimate the disturbance as follows:
\begin{eqnarray}\label{eq_disturbanceobserver}
\dot{p} &=& - \textbf{l} \textbf{z} p - \textbf{l} \Big( \textbf{z} \textbf{l} \textbf{x} + \textbf{g}_{1}(\textbf{x}) + \textbf{g}_{2}(\textbf{x}) u \Big) \nonumber \\
\hat{d}_{BN} &=& p + \textbf{l}\textbf{x}
\end{eqnarray}
where $\hat{d}_{BN}$ denotes the disturbance estimation, $p$ and $\textbf{l}=[\begin{array}{cc} l_{1}, & l_{2} \end{array}]$ denote the internal state and observer gain vector of the BNDO, respectively.

The time-derivative of the disturbance estimation is derived from 
\eqref{eq_disturbanceobserver} as follows:
\begin{equation}\label{eq_do_estimationsignal}
\dot{\hat{d}}_{BN} = \textbf{l} \textbf{z } e_{d}
\end{equation}
where $e_{d}= d - \hat{d}_{BN}$ denotes the estimation error for the disturbance. Since $\textbf{z}=[\begin{array}{cc} 1 & 0 \end{array}]^{T}$ in \eqref{eq_nonlinearsystem}, it is obtained as follows:
 \begin{equation}\label{eq_do_estimationsignal2}
\dot{\hat{d}}_{BN} = l_{1} e_{d}
\end{equation}

The error dynamics of the BNDO are obtained by adding the disturbance rate $\dot{d}$ into \eqref{eq_do_estimationsignal} as:
\begin{equation}\label{eq_do_error}
\dot{e}_{d} + l_{1}   e_{d} = \dot{d}
\end{equation}

\begin{assumption}\label{assumption_do_1}
The time-derivative of the disturbance is bounded and has a steady-state value, i.e., $\displaystyle \lim_{ t \to \infty} \dot{d}(t) = 0$.
\end{assumption}

If Assumption \ref{assumption_do_1} is satisfied, then \eqref{eq_do_error} is obtained as follows:
\begin{eqnarray} \label{eq_do_error_2}
\dot{e}_{d}  + l_{1}  e_{d}  = 0
\end{eqnarray}

\begin{lemma}\label{lemma_do_1}
\cite{chen2003nonlinear}: If $l_{1}$ is positive, i.e., $l_{1}>0$, the error dynamics of the BNDO in \eqref{eq_do_error_2} converge to zero exponentially. 
\end{lemma}

Lemma \ref{lemma_do_1} implies that the estimated disturbance by the BNDO can track the actual disturbance of the system in \eqref{eq_nonlinearsystem} exponentially if the disturbance has a steady-state value. 

\subsubsection{Controller Design}

\begin{figure}[t!]
  \centering
  \includegraphics[width=0.85\columnwidth]{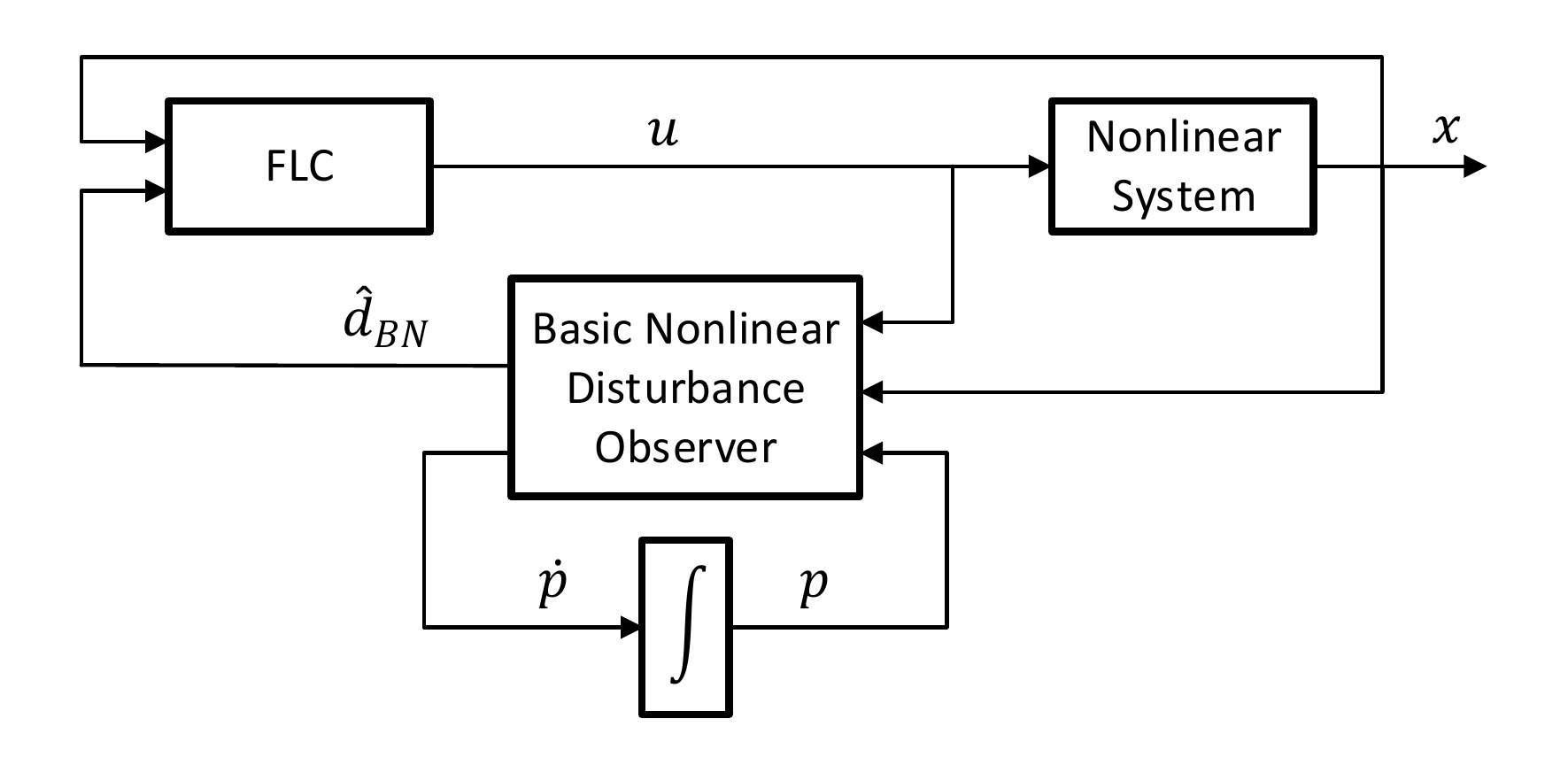}\\
  \caption{The diagram of the disturbance observer (DO) }\label{fig_do_diagram}
\end{figure}

A novel FLC law based on the BNDO to handle mismatched time-invariant uncertainties as shown in Fig. \ref{fig_do_diagram} is formulated as follows:
\begin{equation} \label{eq_flcdo_controllaw}
u = -b^{-1} (\textbf{x}) \Big( a(\textbf{x}) +  k_{1} x_{1} + k_{2} (x_{2} + \hat{d}_{BN}) \Big)
\end{equation}
If the aforementioned control law in \eqref{eq_flcdo_controllaw} is applied to the second-order nonlinear system in \eqref{eq_nonlinearsystem}, the closed-loop dynamics are obtained as
\begin{equation}\label{eq_flcdo_error}
\ddot{x}_{1} + k_{2} \dot{x}_{1} + k_{1} x_{1} = k_{2} e_{d} + \dot{d}
\end{equation}

If Assumption \ref{assumption_do_1} is satisfied, the closed-loop dynamics are obtained as follows:
\begin{equation}\label{eq_flcdo_error_2}
\ddot{x}_{1} + k_{2} \dot{x}_{1} + k_{1} x_{1} = k_{2} e_{d}
\end{equation}

\begin{lemma}\label{lemma_do_2}
\cite{khalil1996nonlinear}: If a nonlinear system $\dot{\textbf{x}}=F(\textbf{x},\textbf{u})$ is input-state stable and $\displaystyle \lim_{ t \to \infty} \textbf{u}(t) = 0$, then the state $\displaystyle \lim_{ t \to \infty} \textbf{x}(t) = 0$.
\end{lemma}

The closed-loop dynamics in \eqref{eq_flcdo_error_2} are stable if the coefficients are positive, i.e., $k_{1}, k_{2}>0$. As indicated in Lemma \ref{lemma_do_2},  if the input satisfies $\displaystyle \lim_{ t \to \infty} e_{d}(t) = 0$, then the state satisfies $\displaystyle \lim_{ t \to \infty} x_{1}(t) = 0$. 

The closed-loop error dynamics in \eqref{eq_do_error_2} and \eqref{eq_flcdo_error_2} are combined as follows:
\begin{eqnarray}\label{eq_flcdo_system}
\dot{e}_{d}  + l_{1}   e_{d} & = & 0 \nonumber \\
\ddot{x}_{1} + k_{2} \dot{x}_{1} + k_{1} x_{1} & = & k_{2} e_{d}
\end{eqnarray}
The closed-loop error dynamics are globally exponentially stable under the given condition $k_{1}, k_{2}, l_{1}>0$ and the states satisfy $\displaystyle \lim_{ t \to \infty} x_{1}(t) = 0$ and $\displaystyle \lim_{ t \to \infty} e_{d}(t) = 0$.  This implies that the states can be driven to the desired equilibrium point. 

\begin{remark} \label{remark_do_1}
The FLC based on the BNDO (FLC-BNDO) is robust to mismatched time-invariant uncertainties. If there exists no disturbance, the FLC-BNDO maintains the nominal control performance as distinct from the FLC-I. 
\end{remark}

\begin{remark} \label{remark_do_2}
If the time-derivative of the disturbance $\dot{d}$ is time-varying, the error dynamics of the BNDO cannot converge to zero. Therefore, there exists always a difference between the estimated and true values of the disturbance so that the FLC-BNDO cannot be robust to mismatched time-varying uncertainties.
\end{remark}


\section{NEW FLC FOR MISMATCHED TIME-VARVARYING UNCERTAINTIES}\label{sec_flcsldo}

\subsection{Self-Learning Disturbance Observer}\label{sec_SLDO}

An SLDO is developed in this section due to the fact that the BNDO \eqref{eq_disturbanceobserver} cannot estimate time-varying disturbances as stated in Remark \ref{remark_do_2}. In this investigation, feedback-error learning (FEL) scheme in which the conventional estimation law works in parallel with a T2NFS is proposed as the diagram of the SLDO is illustrated in Fig \ref{fig_sldo_diagram}. The new estimation law is proposed as follows: 
\begin{equation}\label{eq_sldo_estimationlaw}
\tau = \tau_{c} - \tau_{n}  
\end{equation}
where $\tau_{c}$ and $\tau_{n}$ denote respectively the outputs of the conventional estimation law and T2NFS. A conventional proportional-derivative estimation law in FEL scheme is used to guarantee the global asymptotic stability of the SLDO in a compact space. The input to FEL algorithm is the output of the BNDO and the time-derivative of its output $\dot{\hat{d}}_{BN}$ is equal to the multiplication of the disturbance estimation error and the observer gain of the BNDO as can be seen from \eqref{eq_do_estimationsignal2}. Therefore, the conventional estimation law is formulated as follows:
\begin{equation}\label{eq_sldo_tauc}
\tau_{c} = \dot{\hat{d}}_{BN}  +  \eta \ddot{\hat{d}}_{BN} 
\end{equation}
where $\eta$ is the coefficient and positive, i.e., $\eta>0$, and $d_{BN}$ is the output of the BNDO. 
\begin{figure}[t!]
  \centering
  \includegraphics[width=\columnwidth]{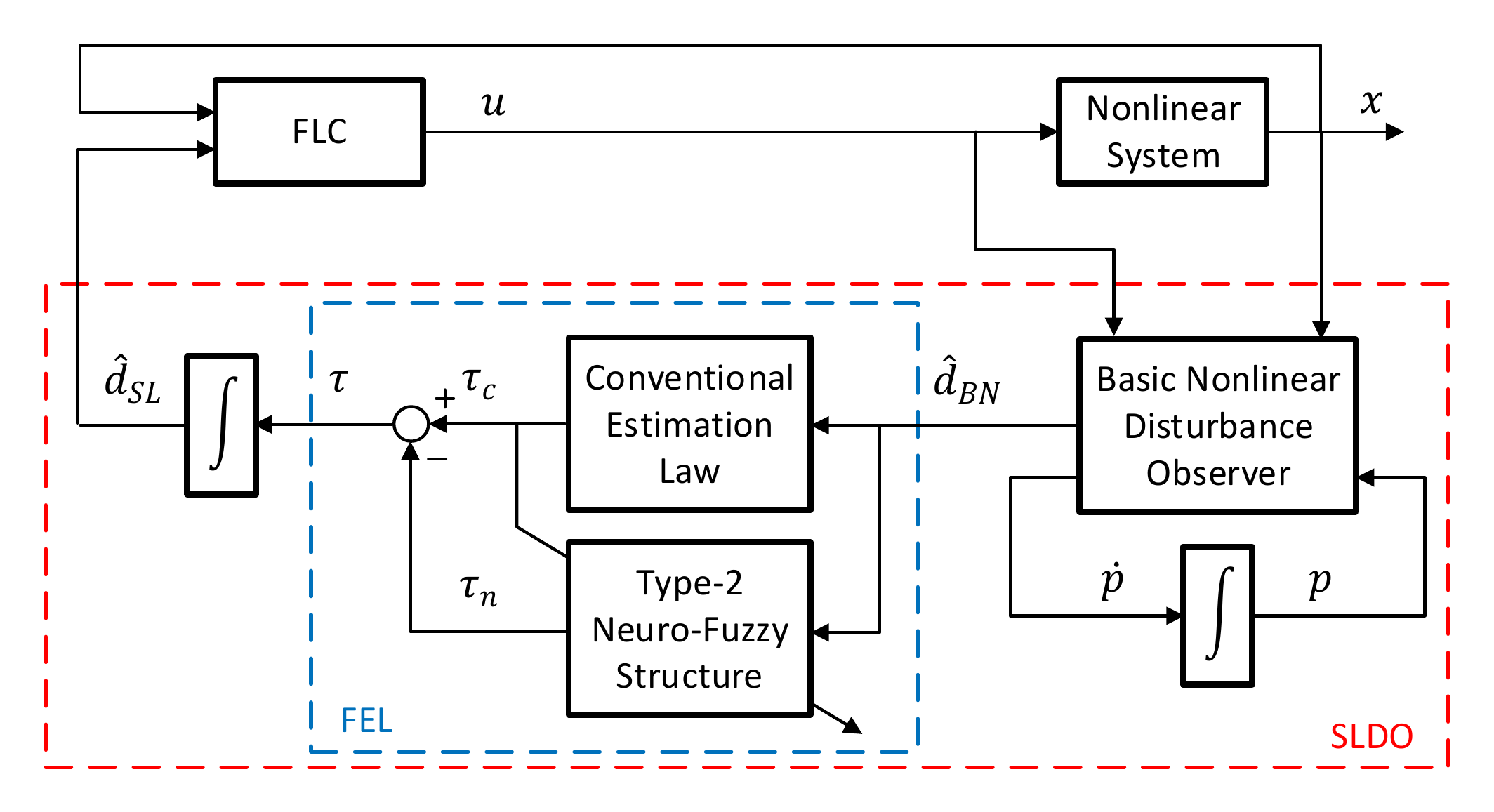}\\
  \caption{The diagram of the self-learning disturbance observer (SLDO) }\label{fig_sldo_diagram}
\end{figure}

\subsubsection{Type-2 Neuro-Fuzzy Structure}\label{sec_neuro-fuzzy}

An interval type-2 Takagi-Sugeno-Kang fuzzy \emph{if-then} rule $R_{ij}$ is written as:
\begin{equation}\label{eq_Rlinearfunction}
R_{ij}: \;\; \textrm{If} \; \xi_1 \; \textrm{is} \;\; \widetilde{1}_i \;\; \textrm{and} \; \xi_2 \; \textrm{is} \;\; \widetilde{2}_j, \;\; \textrm{then} \; f_{ij}=\Upsilon_{ij}
\end{equation}
where $\xi_1=\dot{\hat{d}}_{BN} $ and $\xi_2= \ddot{\hat{d}}_{BN}$ denote the inputs while $\widetilde{1}_i$ and $\widetilde{2}_j$  denote type-2 fuzzy sets for inputs. The function $f_{ij}$ is the output of the rules and the total number of the rules are equal to $K=I \times J$ in which $I$ and $J$ are the total number of the  inputs.

The upper and lower Gaussian membership functions for type-2 fuzzy logic systems are written as follows:
\begin{eqnarray}\label{eq_mu1_lower}
\underline{\mu}_{1i}(\xi_1) = \exp\Bigg(-\bigg(\frac{\xi_{1}-\underline{c}_{1i}}{\underline{\sigma}_{1i}}\bigg)^2\Bigg)
\\ \label{eq_mu1_upper}
\overline{\mu}_{1i}(\xi_1) = \exp\Bigg(-\bigg(\frac{\xi_{1}-\overline{c}_{1i}}{\overline{\sigma}_{1i}}\bigg)^2\Bigg)
\\ \label{eq_mu2_lower}
\underline{\mu}_{2j}(\xi_2) = \exp\Bigg(-\bigg(\frac{\xi_{2}-\underline{c}_{2j}}{\underline{\sigma}_{2j}}\bigg)^2\Bigg)
\\ \label{eq_mu2_upper}
\overline{\mu}_{2j}(\xi_2) = \exp\Bigg(-\bigg(\frac{\xi_{2}-\overline{c}_{2j}}{\overline{\sigma}_{2j}}\bigg)^2\Bigg)
\end{eqnarray}
where $\underline{c}, \overline{c}, \underline{\sigma}, \overline{\sigma}$ denote respectively the lower and upper mean, and the lower and upper standard deviation of the membership functions. These parameters are adjustable for the T2NFS in which the standard deviations are positive, i.e., $ \underline{\sigma}, \overline{\sigma} > 0$.

The lower and upper membership functions  $\underline{\mu }$ and $\overline{\mu }$ are determined for every signal. Then, the firing strength of rules are calculated as follows:
\begin{equation}\label{eq_wij_lower_upper}
\underline{w}_{ij} = \underline{\mu}_{1i}(\xi_1)  \underline{\mu}_{2j}(\xi_2), \quad 
\overline{w}_{ij} = \overline{\mu}_{1i}(\xi_1) \overline{\mu}_{2j}(\xi_2)
\end{equation}

The output of the every fuzzy rule is a linear function $f_{ij}$ formulated in \eqref{eq_Rlinearfunction}.  The output of the network is formulated below:
\begin{equation}\label{eq_taun}
\tau_n=q \sum_{i=1}^{I}\sum_{j=1}^{J}f_{ij}\widetilde{\underline{w}}_{ij}+(1-q)\sum_{i=1}^{I}\sum_{j=1}^{J}f_{ij}\widetilde{\overline{w}}_{ij}
\end{equation}
where $\widetilde{\underline{{w}}}_{ij}$ and $\widetilde{\overline{{w}}}_{ij}$ are the normalized firing strengths of the lower and upper output signals of the neuron $ij$ are written as follows:
\begin{equation}\label{eq_wij_lower_upper_normalized}
\widetilde{\underline{w}}_{ij} = \frac{\underline{w}_{ij}}{\sum_{i=1}^{I}\sum_{j=1}^{J}\underline{w}_{ij}}, \quad
\widetilde{\overline{w}}_{ij} = \frac{\overline{w}_{ij}}{\sum_{i=1}^{I}\sum_{j=1}^{J}\overline{w}_{ij}}
\end{equation}
The design parameter $q$ weights the participation of the lower and upper firing levels and is generally set to $0.5$. In this paper, it is formulated as a time-varying parameter in the next subsection. 

The vectors are defined as::
\begin{eqnarray}
\widetilde{\underline{\textbf{W}}}&=& [\widetilde{\underline{w}}_{11}\; \widetilde{\underline{w}}_{12} \dots \widetilde{\underline{w}}_{21}\dots \widetilde{\underline{w}}_{ij} \dots \widetilde{\underline{w}}_{IJ}]^{T} \nonumber \\
\widetilde{\overline{\textbf{W}}}&=& [\widetilde{\overline{w}}_{11} \; \widetilde{\overline{w}}_{12} \dots \widetilde{\overline{w}}_{21} \dots \widetilde{\overline{w}}_{ij} \dots \widetilde{\overline{w}}_{IJ}]^{T} \nonumber \\
\textbf{F} &=& [f_{11} \; f_{12} \dots f_{21} \dots f_{ij} \dots f_{IJ}] \nonumber
\end{eqnarray}
where these normalized firing strengths are between $0$ and $1$, i.e., $0<\widetilde{\underline{w}}_{ij} \leq 1$ and $0<\widetilde{\overline{w}}_{ij} \leq 1$. In addition, $\sum_{i=1}^{I}\sum_{j=1}^{J}\widetilde{\underline{w}}_{ij} = 1$ and $\sum _{i=1}^{I}\sum_{j=1}^{J}\widetilde{\overline{w}}_{ij} = 1$.

\subsubsection{SMC Theory-based Learning Algorithm}\label{sec_SMClearning}

In FEL algorithm, the output of the conventional estimation law $\tau_{c}$ must converge to zero while T2NFS is learning; therefore, the proposed sliding surface $s$ for the learning algorithm contains the conventional estimation law and is formulated as follows:
\begin{equation}\label{eq_sldo_slidingfsurface}
s = \tau_{c} +\frac{ \ddot{\hat{d}}_{BN} } {l_{1}}
\end{equation}
where $s$ is used as the learning error to train the learning algorithm.

Novel adaptation rules for the T2NFS parameters are proposed in this investigation by the following equations:
\begin{equation} \label{eq_c_1i}
\dot{\underline{c}}_{1i} = \dot{\overline{c}}_{1i} = \dot{\xi}_{1} + \xi_{1} \alpha \textrm{sgn}\left( s \right)
\end{equation}
\begin{equation} \label{eq_c_2j}
\dot{\underline{c}}_{2j} = \dot{\overline{c}}_{2j} = \dot{\xi}_{2} + \xi_{2}\alpha \textrm{sgn}\left( s  \right)
\end{equation}
\begin{equation}\label{eq_sigma_1i_lower}
\dot{\underline{\sigma}}_{1i} = - \frac{\underline{\sigma}_{1i}}{\xi_{1} - \underline{c}_{1i}} \bigg( \xi_{1} + \frac{ (\underline{\sigma}_{1i} )^2}{\xi_{1} - \underline{c}_{1i}} \bigg) \alpha  \textrm{sgn}\left( s  \right)
\end{equation}
\begin{equation}\label{eq_sigma_1i_upper}
\dot{\overline{\sigma}}_{1i} = - \frac{\overline{\sigma}_{1i}}{\xi_{1} - \overline{c}_{1i}} \bigg( \xi_{1} + \frac{ (\overline{\sigma}_{1i} )^2}{\xi_{1} - \overline{c}_{1i}} \bigg) \alpha  \textrm{sgn}\left(  s \right)
\end{equation}
\begin{equation}\label{eq_sigma_2j_lower}
\dot{\underline{\sigma}}_{2j} = - \frac{\underline{\sigma}_{2j}}{\xi_{2} - \underline{c}_{2j}} \bigg( \xi_{2} + \frac{ (\underline{\sigma}_{2j} )^2}{\xi_{2} - \underline{c}_{2j}} \bigg) \alpha  \textrm{sgn}\left(  s \right)
\end{equation}
\begin{equation}\label{eq_sigma_2j_upper}
\dot{\overline{\sigma}}_{2j} = - \frac{\overline{\sigma}_{2j}}{\xi_{2} - \overline{c}_{2j}} \bigg( \xi_{2} + \frac{ (\overline{\sigma}_{2j} )^2}{\xi_{2} - \overline{c}_{2j}} \bigg) \alpha  \textrm{sgn}\left(  s  \right)
\end{equation}
\begin{equation}\label{eq_f_ij}
\dot{f}_{ij} =-\frac{q \widetilde{\underline{w}}_{ij}+ (1-q)\widetilde{\overline{w}}_{ij}}{(q\widetilde{\underline{\textbf{W}}}+(1-q) \widetilde{\overline{\textbf{W}}})^T(q\widetilde{\underline{\textbf{W}}}+ (1-q)\widetilde{\overline{\textbf{W}}})}\alpha sgn( s )
\end{equation}
\begin{equation}\label{eq_q}
\dot{q} =-\frac{1}{\textbf{F}(\widetilde{\underline{\textbf{W}}}-\widetilde{\overline{\textbf{W}}})^{T}}\alpha sgn(  s  )
\end{equation}
where $\alpha$ is the learning rate and positive, i.e., $\alpha>0$.

\begin{theorem}[Stability of the learning algorithm]\label{theorem1}
If adaptations rules are proposed as in \eqref{eq_c_1i}-\eqref{eq_q} and the learning rate $\alpha$ is large enough, i.e., $\alpha > \dot{\tau}^{*}  \frac{ \mid\tau_{c} + \dot{e}_{d} \mid + \mid \dot{e}_{d} \mid}{ \mid\tau_{c} + \dot{e}_{d} \mid - \mid \dot{e}_{d} \mid}$ where $\dot{\tau}^{*}$ is the upper bound of $\dot{\tau}$, this ensures that $\tau_{c}$ will converge to zero in finite time for a given arbitrary initial condition $\tau_{c}(0)$.  
\end{theorem}

\begin{proof}
The Lyapunov function is written as follows:
\begin{equation}\label{eq_sldo_Lyapunov}
V = \frac{1}{2} \tau_{c}^{2} 
\end{equation}
By taking the time-derivative of the Lyapunov function in \eqref{eq_sldo_Lyapunov}, it is obtained as
\begin{equation}\label{eq_sldo_Lyapunov_dot} 
\dot{V} =   \tau_{c} \dot{\tau}_{c} 
\end{equation}
If the time-derivative of the conventional estimation law $\dot{\tau}_{c}$ is derived from  \eqref{eq_sldo_estimationlaw} and then is inserted into the equation above:
\begin{equation}\label{eq_sldo_Lyapunov_dot2} 
\dot{V} =   \tau_{c} (\dot{\tau}_{n} + \dot{\tau} ) 
\end{equation}
The calculation of $\dot{\tau}_{n}$ in \eqref{dotVc4} is inserted into \eqref{eq_sldo_Lyapunov_dot2}, it is obtained as follows:
\begin{equation}\label{eq_sldo_Lyapunov_dot4}
\dot{V} =   \tau_{c} \Big( -2 \alpha  \textrm{sgn}\left( s \right)  +  \dot{\tau}  \Big) 
\end{equation}
The sliding surface in \eqref{eq_sldo_slidingfsurface} is inserted into the aforementioned equation
\begin{equation}\label{eq_sldo_Lyapunov_dot5}
\dot{V} =   \tau_{c} \Bigg( -2 \alpha  \textrm{sgn}\left( \tau_{c} +\frac{ \ddot{\hat{d}}_{BN} } {l_{1}} \right)  +  \dot{\tau}  \Bigg) 
\end{equation}
It is obtained considering \eqref{eq_do_estimationsignal2} as follows:
\begin{eqnarray}\label{eq_sldo_Lyapunov_dot6}
\dot{V} &=&   \tau_{c} \Big( -2 \alpha  \textrm{sgn}\left( \tau_{c} +\dot{e}_{d} \right)  +  \dot{\tau}  \Big)  \nonumber \\
&=&  \Big( \tau_{c} + \dot{e}_{d} \Big) \Big( -2 \alpha  \textrm{sgn}\left( \tau_{c} +\dot{e}_{d} \right)  +  \dot{\tau}  \Big)  \nonumber \\
 && - \dot{e}_{d}  \Big( -2 \alpha  \textrm{sgn}\left( \tau_{c} +\dot{e}_{d} \right)  +  \dot{\tau}  \Big) 
\end{eqnarray}

If it is assumed that $\dot{\tau}$ is upper bounded by $\dot{\tau}^{*}$, \eqref{eq_sldo_Lyapunov_dot6} is obtained as follows:
\begin{eqnarray}\label{eq_sldo_Lyapunov_dot8}
\dot{V} &<&  \mid\tau_{c} + \dot{e}_{d} \mid ( -2 \alpha  +  \dot{\tau}^{*}  ) + \mid \dot{e}_{d} \mid ( 2 \alpha +  \dot{\tau}^{*}  )  \nonumber \\
&<&  - 2 \alpha \Big( \mid\tau_{c} + \dot{e}_{d} \mid - \mid \dot{e}_{d} \mid \Big) \nonumber \\
&& + \dot{\tau}^{*} \Big( \mid\tau_{c} + \dot{e}_{d} \mid + \mid \dot{e}_{d} \mid \Big) \nonumber \\
&< & - 2 \alpha + \dot{\tau}^{*}  \frac{ \mid\tau_{c} + \dot{e}_{d} \mid + \mid \dot{e}_{d} \mid}{ \mid\tau_{c} + \dot{e}_{d} \mid - \mid \dot{e}_{d} \mid}
\end{eqnarray}

As stated in Theorem \ref{theorem1}, if the learning rate $\alpha$ is large enough, i.e., $\alpha> \dot{\tau}^{*}  \frac{ \mid\tau_{c} + \dot{e}_{d} \mid + \mid \dot{e}_{d} \mid}{ \mid\tau_{c} + \dot{e}_{d} \mid - \mid \dot{e}_{d} \mid}$, then the time-derivative of the Lyapunov function is negative, i.e., $\dot{V}<0$, so that the SMC theory-based learning algorithm is globally asymptotically stable and $\tau_{c}$ will converge to zero in finite time.
\end{proof}

%

\subsubsection{Stability of SLDO}\label{sec_stabilitySLDO}

The proposed SLDO law in \eqref{eq_sldo_estimationlaw} is re-written taking \eqref{eq_sldo_tauc} into account as follows:
\begin{eqnarray}\label{eq_obs_estimation_law_self_stability_ilk}
\tau =  \dot{\hat{d}}_{BN}  + \eta \ddot{\hat{d}}_{BN}  - \tau_{n}
\end{eqnarray}

The error dynamics for the SLDO are obtained by adding the actual disturbance rate $\dot{d}$ into the estimated disturbance rate in \eqref{eq_obs_estimation_law_self_stability_ilk} and considering the dynamics of the time-derivative of the estimated disturbance by BNDO in \eqref{eq_do_estimationsignal} and $\tau=\dot{\hat{d}}_{SL}$:
\begin{eqnarray}\label{eq_obs_estimation_law_self_stability}
 \dot{d} - \dot{\hat{d}}_{SL}  &=&  - \dot{\hat{d}}_{BN}  - \eta  \ddot{\hat{d}}_{BN}  + \tau_{n}+ \dot{d} \nonumber \\
\dot{e}_{d}  &=& - l_{1} e_{d} - \eta l_{1} \dot{e}_{d}   + \tau_{n} + \dot{d} 
\end{eqnarray}

By taking the time-derivative of \eqref{eq_obs_estimation_law_self_stability}, it is obtained as follows:
\begin{equation}\label{eq_obs_error_self_2}
\ddot{e}_{d}  =   \frac{  -l_{1} \dot{e}_{d}  + \dot{\tau}_{n} + \ddot{d} } { 1+ \eta l_{1} }
\end{equation}
As calculated in Appendix, $\dot{\tau}_{n}=- 2\alpha sgn(s)$ is inserted into \eqref{eq_obs_error_self_2};
\begin{equation}\label{eq_obs_error_self_3}
\ddot{e}_{d}  =   \frac{  -l_{1} \dot{e}_{d}   - 2\alpha sgn(s) + \ddot{d}} {1 + \eta l_{1}}
\end{equation}

If  $\tau_{c}$ in \eqref{eq_sldo_tauc} is inserted into \eqref{eq_sldo_slidingfsurface}, the sliding surface is obtained as follows:
\begin{eqnarray}\label{eq_slidingfsurface_2}
s \left( \dot{\hat{d}}_{BN}, \ddot{\hat{d}}_{BN} \right)  = ( \eta + \frac{1}{l_{1}} ) \ddot{\hat{d}}_{BN} + \dot{\hat{d}}_{BN} 
\end{eqnarray}

\begin{theorem}[Stabiltiy of the SLDO]\label{theorem2}
The estimation law in \eqref{eq_sldo_estimationlaw} is employed as a DO, the closed-loop error dynamics for the SLDO are asymptotically stable if the learning rate of the T2NFS $\alpha$ is large enough, i.e., $\alpha > \ddot{d}^{*} $ where the acceleration of the actual disturbance $\ddot{d}$ is upper bounded by $\ddot{d}^{*}$.
\end{theorem}

\begin{proof}
The Lyapunov function is written as follows:
\begin{equation}\label{eq_observer_Lyapunov}
V= \frac{1}{2}  s^{2} 
\end{equation}
By taking the time-derivative of the Lyapunov function above, it is obtained as
\begin{eqnarray}\label{eq_observer_Lyapunov_d} 
\dot{V} =  s \dot{s} 
\end{eqnarray}
If the time-derivative of the sliding surface in \eqref{eq_slidingfsurface_2} is inserted into the aforementioned equation, it is obtained as follows:  
\begin{eqnarray}\label{eq_observer_Lyapunov_dd} 
 \dot{V} =  s \Big(\dddot{\hat{d}}_{BN} (\eta + \frac{1}{l_{1}}) + \ddot{\hat{d}}_{BN}  \Big)
\end{eqnarray}
The aforementioned equation is obtained considering \eqref{eq_do_estimationsignal2}
\begin{eqnarray}\label{eq_observer_Lyapunov_d1} 
 \dot{V} =  s \Big(  \ddot{e}_{d} (\eta l_{1} + 1 ) + l_{1} \dot{e}_{d}\Big)
\end{eqnarray}
\eqref{eq_obs_error_self_3} is inserted into \eqref{eq_observer_Lyapunov_d1}, it is obtained as follows:
\begin{eqnarray}\label{eq_observer_Lyapunov_d2}
 \dot{V} =  s \Big(  (\frac{  -l_{1} \dot{e}_{d}  - 2 \alpha sgn(s) + \ddot{d}} {1 + \eta l_{1}}) (\eta l_{1} + 1 ) + l_{1} \dot{e}_{d}\Big)
\end{eqnarray}
If it is assumed that $\ddot{d}$ is upper bounded by $\ddot{d}^{*}$:
\begin{eqnarray}\label{eq_observer_Lyapunov_d3}
 \dot{V}  <  \mid s \mid  (- 2 \alpha + \ddot{d}^{*})
\end{eqnarray}

As stated in Theorem \ref{theorem2}, if the learning rate $\alpha$ is large enough, i.e., $\alpha >\ddot{d}^{*}$, then the time-derivative of the Lyapunov function is negative, i.e., $\dot{V}<0$, so that the SLDO is globally asymptotically stable. This purports that the closed-loop error dynamics for the SLDO converge to zero.
\end{proof}

\subsection{Controller Design}

A novel control law based on the SLDO by taking the estimated values of the disturbance and disturbance rate into account are given as follows:
\begin{eqnarray}\label{eq_flcsldo_controllaw}
u = -b^{-1} (\textbf{x})\Big( a(\textbf{x}) +  k_{1} x_{1} + k_{2} (x_{2} + \hat{d}_{SL}) + \dot{\hat{d}}_{SL} \Big)
\end{eqnarray}
If the aforementioned control law in \eqref{eq_flcsldo_controllaw} is applied to the second-order nonlinear system in \eqref{eq_nonlinearsystem}, the closed-loop dynamics are obtained as
\begin{equation}\label{eq_flcsldo_error}
\ddot{x}_{1} + k_{2} \dot{x}_{1} + k_{1} x_{1} = k_{2} e_{d} + \dot{e}_{d}
\end{equation}
where $e_{d}$ and $\dot{e}_{d}$ denote the estimation errors for the disturbance and disturbance rate, respetively.

As stated in Theorem \ref{theorem2}, $e_{d}$ can converge to zero and the closed-loop error dynamics \eqref{eq_flcsldo_error} is stable if the coefficients are positive, i.e., $k_{1}, k_{2}>0$. Therefore, since the input satisfies $\displaystyle \lim_{ t \to \infty} e_{d}(t) = 0$ and the state satisfies $\displaystyle \lim_{ t \to \infty} x_{1}(t) = 0$ as stated in Lemma \ref{lemma_do_2}. 
This implies that the states can be driven to the desired equilibrium point asymptotically under the control law in \eqref{eq_flcsldo_controllaw} based on the SLDO.

\begin{remark} \label{remark_flcsldo}
The FLC based on the SLDO (FLC-SLDO) is robust to mismatched time-varying uncertainties and maintains the nominal control performance in case of no disturbance.
\end{remark}


\section{SIMULATION STUDIES}\label{sec_simulation}

The controllers designed in this paper are applied to a second-order nonlinear system which is formulated as follows
\begin{eqnarray}\label{eq_chaotic}
\dot{x}_{1} & = & x_{2} + d \nonumber \\
\dot{x}_{2} & = &  -x_{1} - x_{2} + 0.3 \cos{( x_{1} )} + e^{x_1} + u 
\end{eqnarray}
where $a(x)=  -x_{1} - x_{2} + 0.3 \cos{( x_{1} )} + e^{x_1}$,  $b(x)=1$ and $z=[1,0]^{T}$ as can be seen from \eqref{eq_nonlinearsystem}.

SMC theory suffers from high-frequency oscillations called \emph{chattering} because of the discontinuous term. Several approaches have been suggested to avoid chattering. One of the popular ways is to replace the sign function by an approximate $sgn(s)= s / (|s| + \delta)$, which mimics a high gain control \cite{7902222}. In this paper, the $sgn$ functions are replaced by the following equation to decrease the chattering effect:
\begin{equation}\label{sgn}
\textrm{sgn}\left( s\right) := \frac{s}{\mid s \mid + 0.05}
\end{equation}

The initial conditions on the states are set to $\textbf{x}(0)=[1,1]^{T}$. The coefficients of the FLCs are respectively set to $k_{1}=3$, $k_{2}=5$ and $k_{i}=3$. The observer gain of the DO is set to $\textbf{l}=[l_{1}, l_{2}]=[5, 0]$, the coefficient $\eta$ in \eqref{eq_sldo_tauc} is set to $\eta=10$ and the learning rate $\alpha$ of the SLDO is set to $0.03$. The initial condition on the parameter $q$ is set to $0.5$. In order to evaluate different controllers in the absence and presence of the mismatched uncertainties, there exists no disturbance imposed on the system at $t=0-20$ second, a step external disturbance $d = 0.5$ is imposed on the system at $t =20-40$ second and a multi-sinusoidal external disturbance $d=0.25+0.15 \times (\sin{(0.5t)} + \sin{(1.5t)} )$ is imposed on the system at $t=40-60$ second. The sampling time for all the simulations in this study is set to $0.001$ second while the number of membership functions are selected as $I = J = 3$. 

The responses of the states $x_{1}, x_{2}$ are respectively shown in Figs. \ref{fig_x1}-\ref{fig_x2}. In the absence of mismatched uncertainties, all FLCs ensure the nominal control performance while FLC-I cannot give the nominal control performance due to the integral action as stated in Remark \ref{remark_flci}. The integral action causes large overshoot, rise time and settling time. In the presence of mismatched time-invariant uncertainties, the FLC-I, FLC-BNDO, and FLC-SLDO ensure the robust control performance while the traditional FLC is sensitive to mismatched uncertainties as stated in Remark \ref{remark_tflc}. Moreover, the FLC-SLDO gives less rise time, settling time and overshoot than the FLC-I and FLC-BNDO. After $t=40$ second, a time-varying disturbance is imposed on the system to evaluate controllers against mismatched time-varying uncertainties. As seen, only the FLC-SLDO gives robust control performance while other controllers fail to drive the states to the desired equilibrium point.

\begin{figure*}[t!]
\centering
\subfigure[ ]{
\includegraphics[width=\columnwidth]{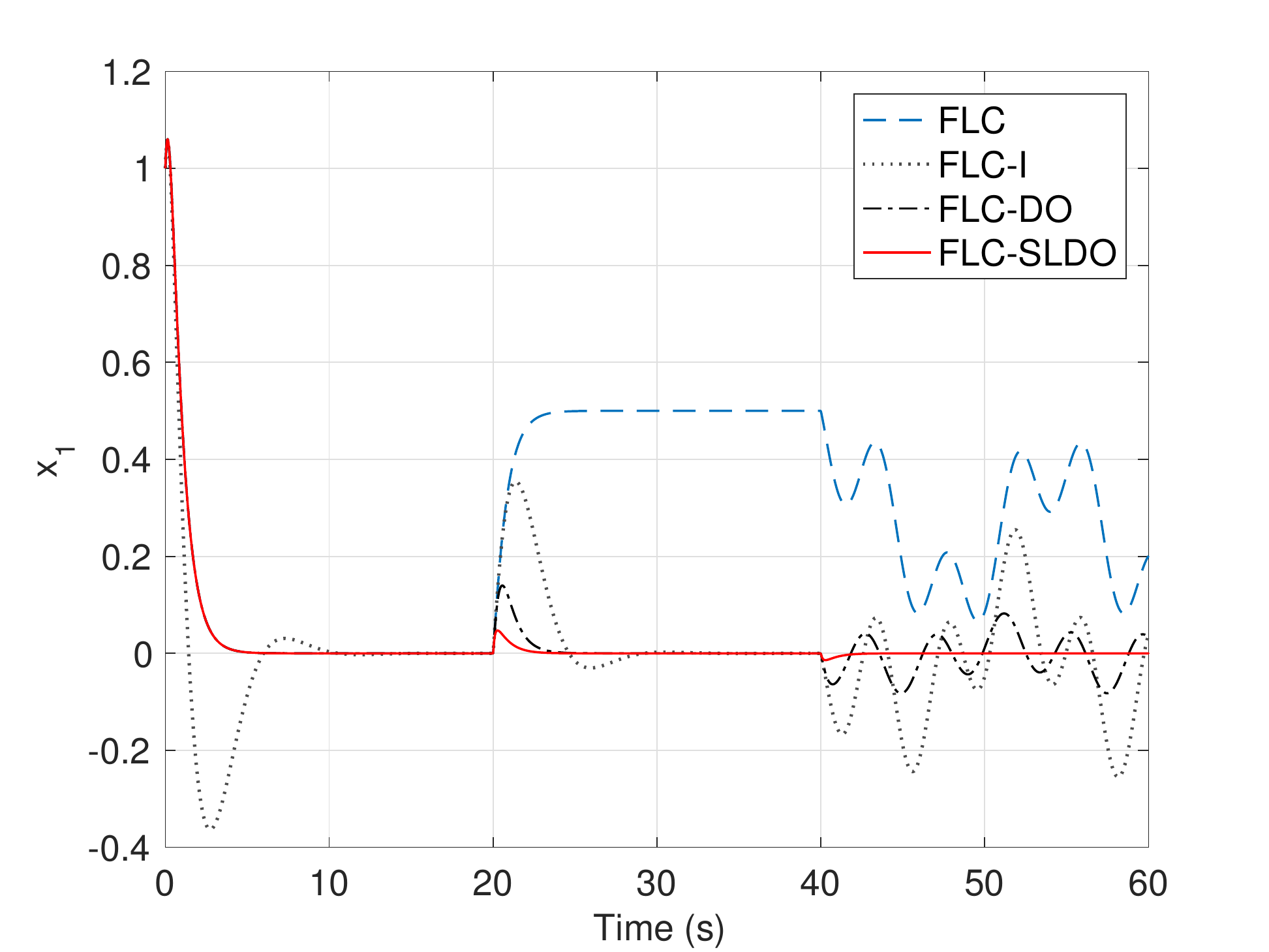}
\label{fig_x1}
}
\subfigure[ ]{
\includegraphics[width=\columnwidth]{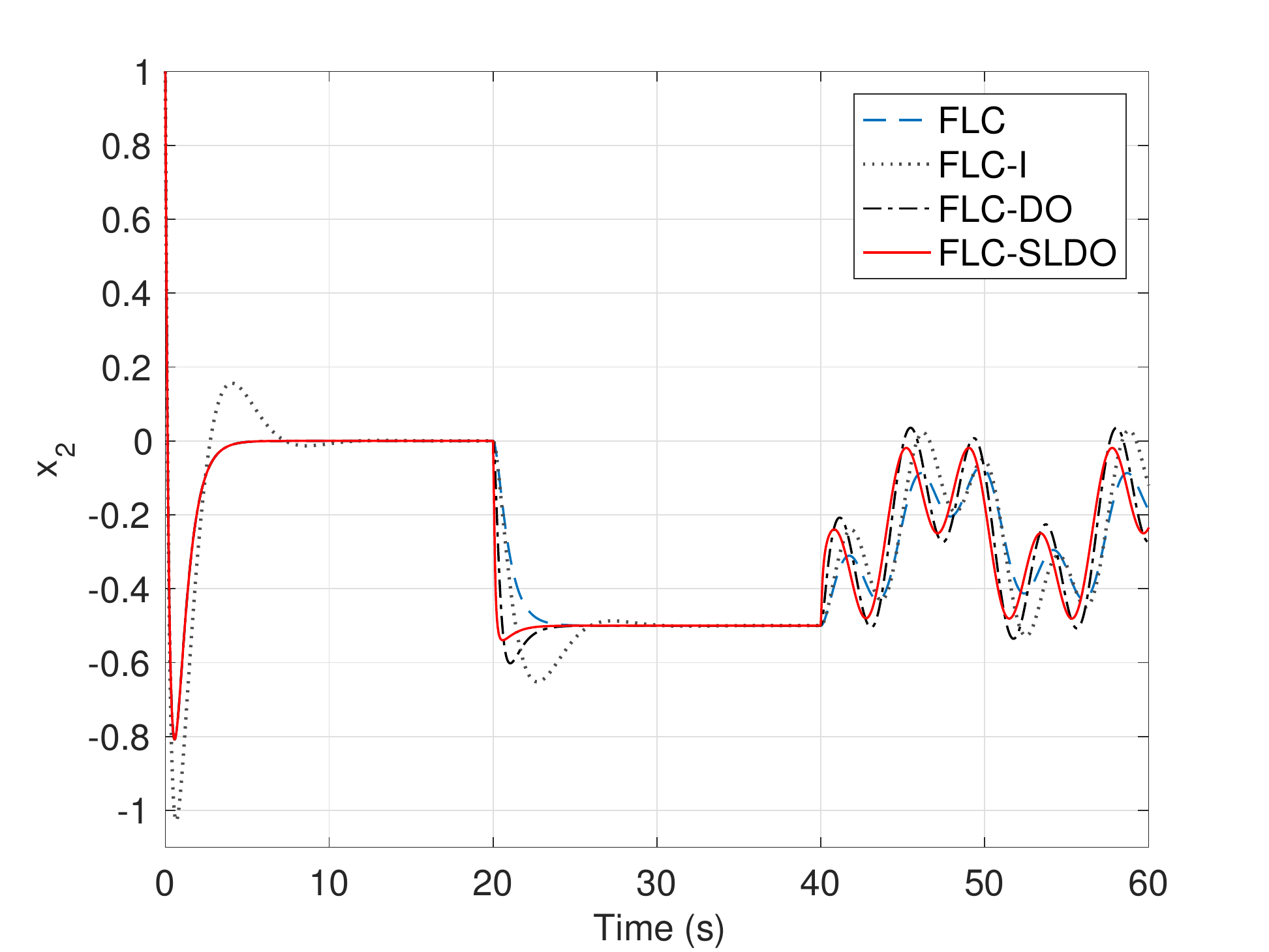}
\label{fig_x2}
}
\subfigure[ ]{
\includegraphics[width=\columnwidth]{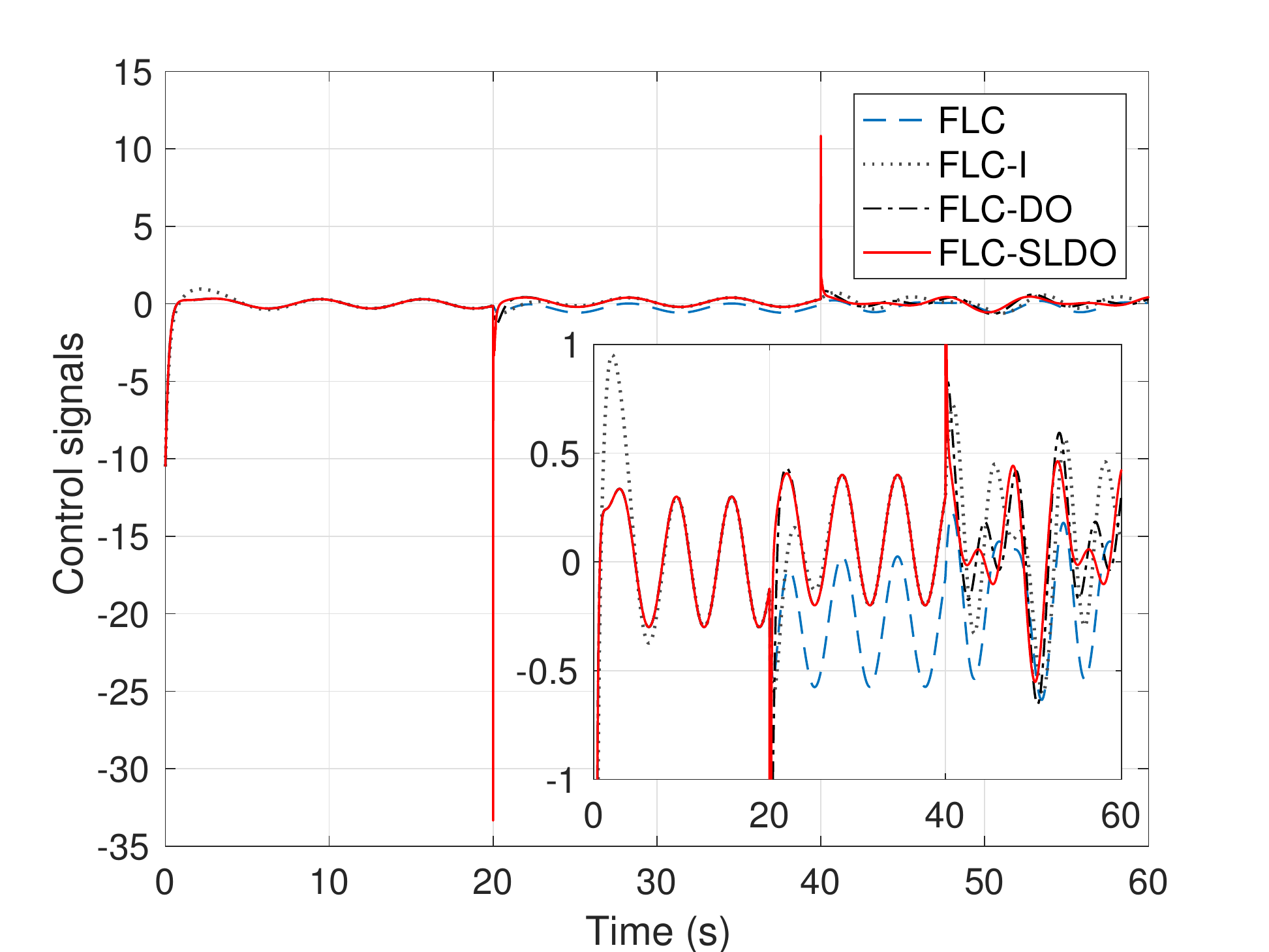}
\label{fig_control_signals}
}
\subfigure[ ]{
\includegraphics[width=\columnwidth]{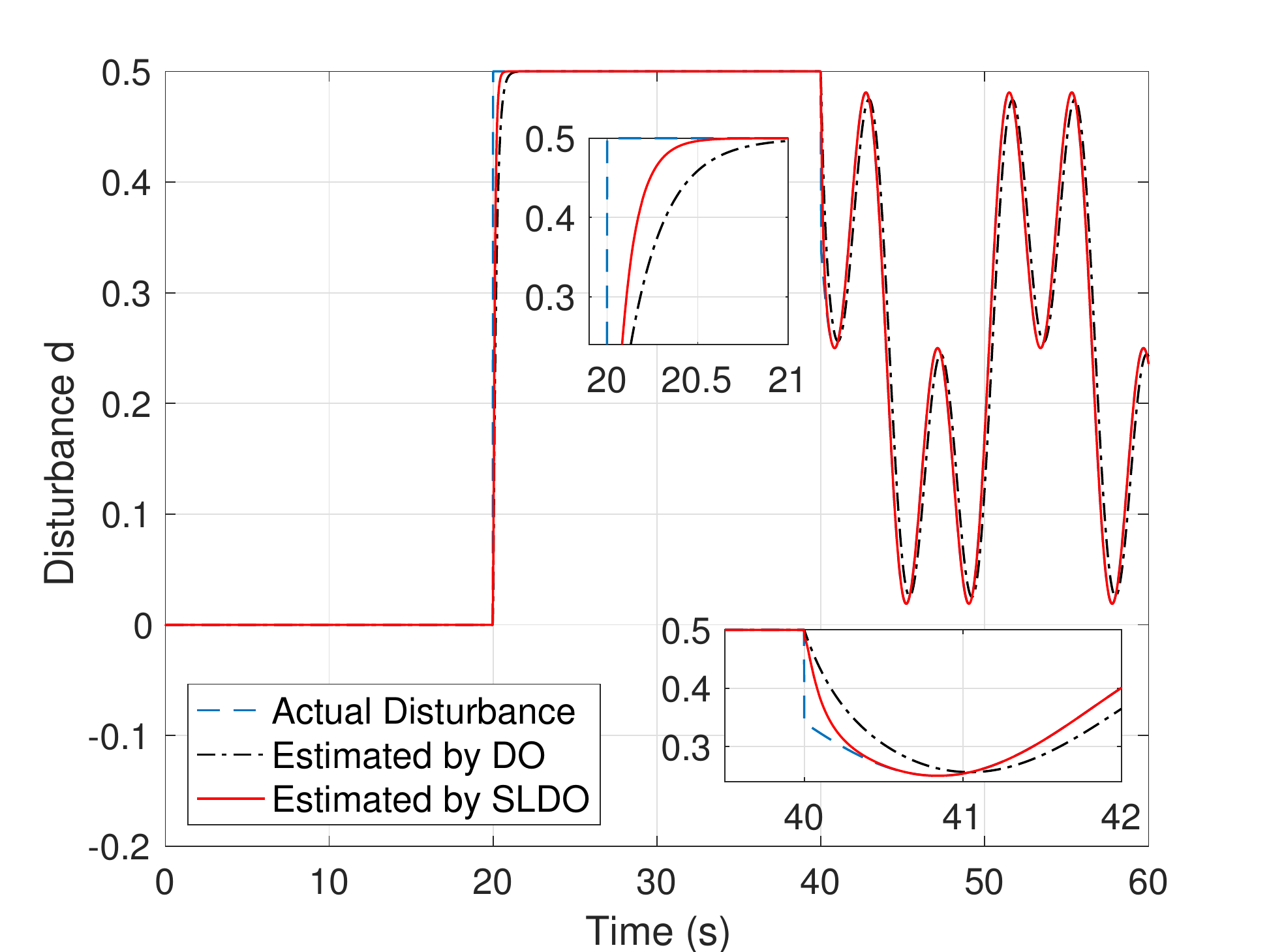}
\label{fig_d}
}
\subfigure[ ]{
\includegraphics[width=\columnwidth]{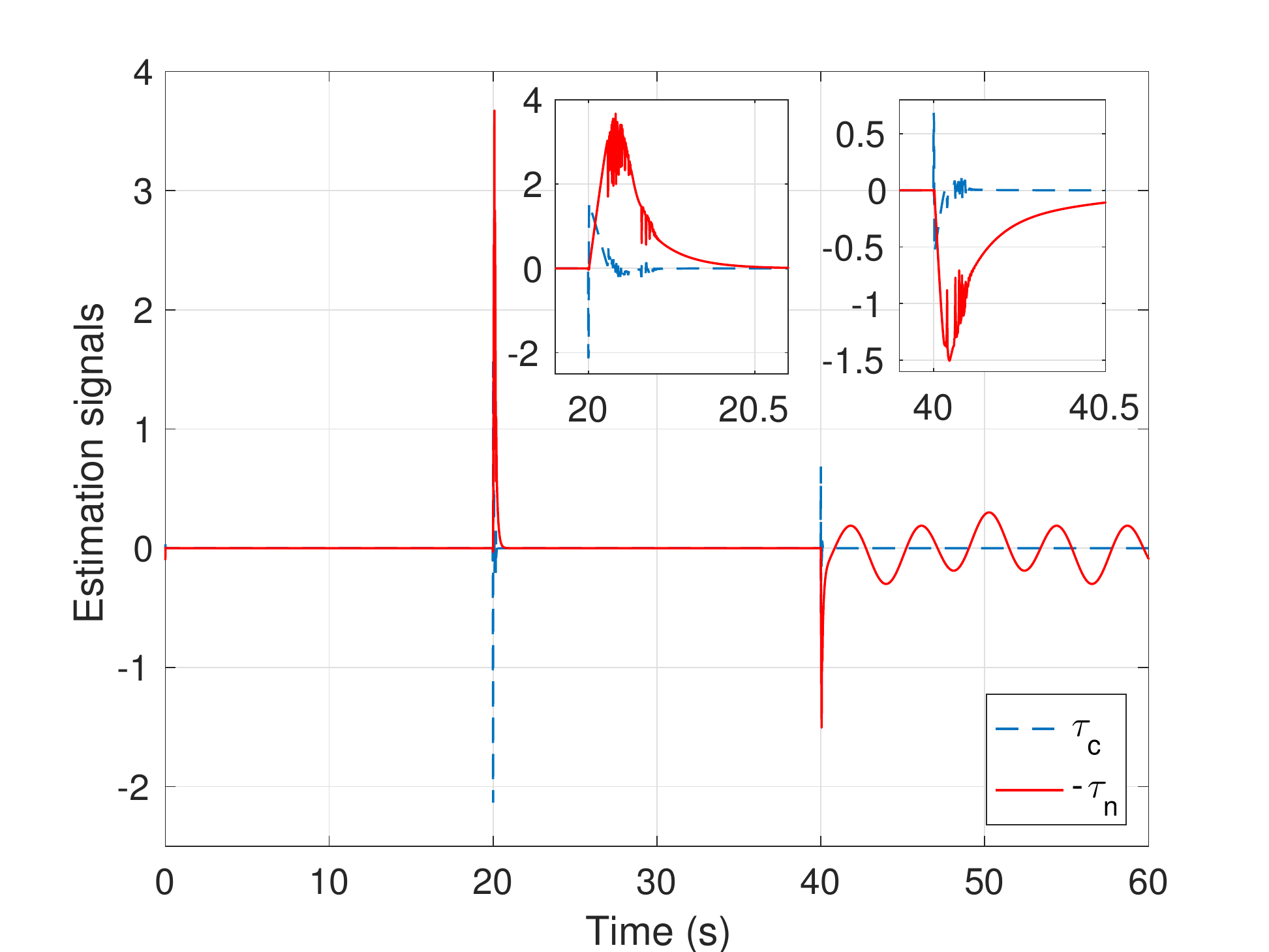}
\label{fig_estimation_signals}
}
\subfigure[ ]{
\includegraphics[width=\columnwidth]{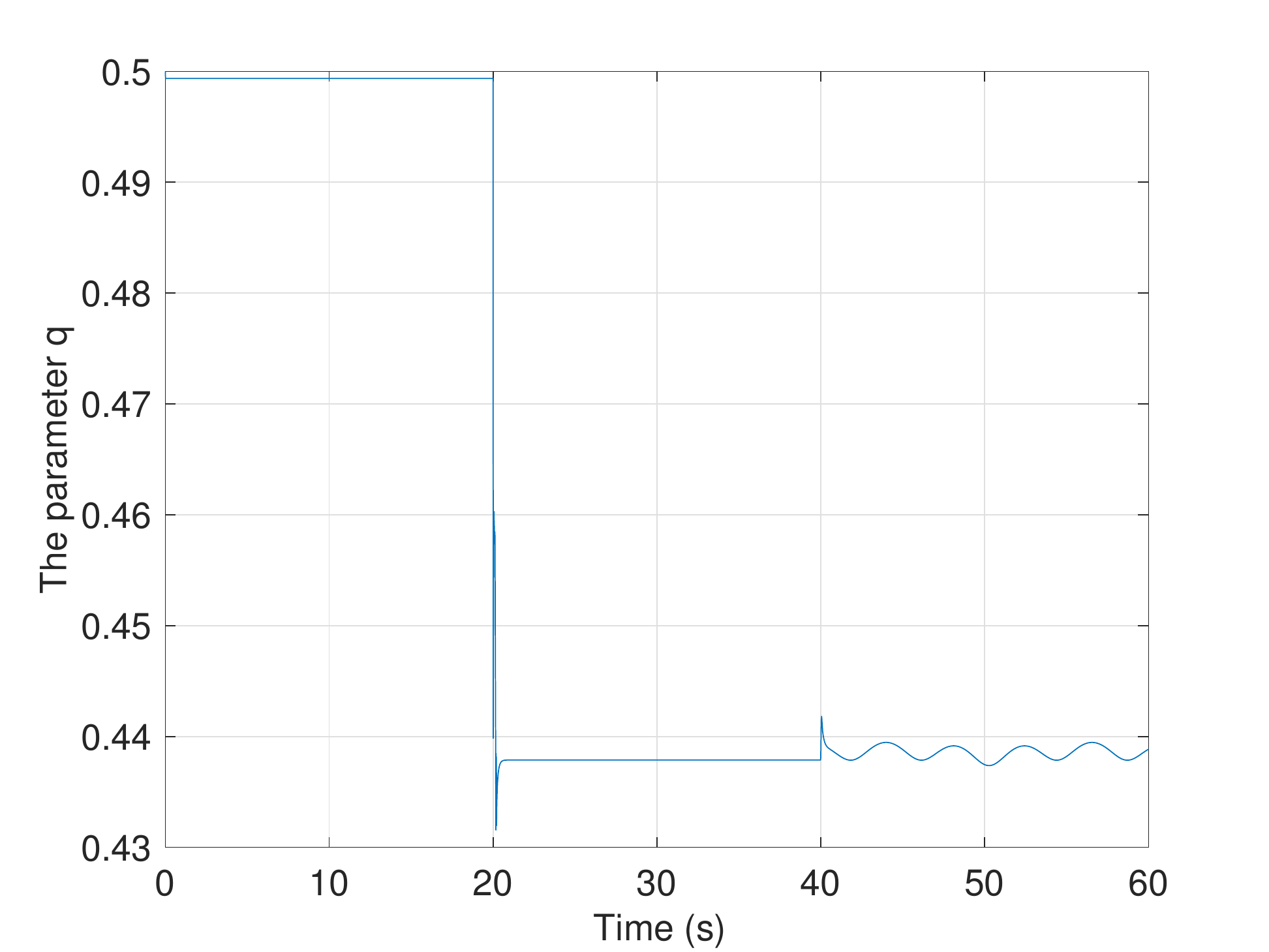}
\label{fig_q}
}
\caption[Optional caption for list of figures]{(a) The responses of state $x_{1}$ (b) The responses of state $x_{2}$ (c) The control signals (d) The true and estimated values of the distubance (e) The estimation signals (f) The parameter q }
\label{sensors}
\end{figure*}

The control signals are shown in Fig. \ref{fig_control_signals}, and the actual and estimated disturbances are shown in Fig. \ref{fig_d}. The BNDO can estimate only time-invariant disturbances while the SLDO can estimate both time-invariant and time-varying disturbances. Moreover, the SLDO reaches the true disturbance value faster than the BNDO in the presence of time-invariant disturbances. Thanks to learning by the T2NFS, the T2NFS can take the overall control of the estimation signal while the conventional estimation signal $\tau_{c}$ converges to zero in finite time as shown in Fig \ref{fig_estimation_signals}. Thus, the T2NFS becomes the leading estimation signal after a short time-period in the SLDO. The adaptation of the parameter $q$ is shown in Fig. \ref{fig_q}. Thanks to the adaptation rule in \eqref{eq_q}, it is adjusted throughout the simulations. Moreover, the phase portraits for all controllers are shown in Figs. \ref{fig_x1dotx1_1}-\ref{fig_x1dotx1_3}. The states driven to the desired equilibrium point can be seen.

%
%
%
%
%
%

\begin{figure*}[t!]
\centering
\subfigure[ ]{
\includegraphics[width=0.66\columnwidth]{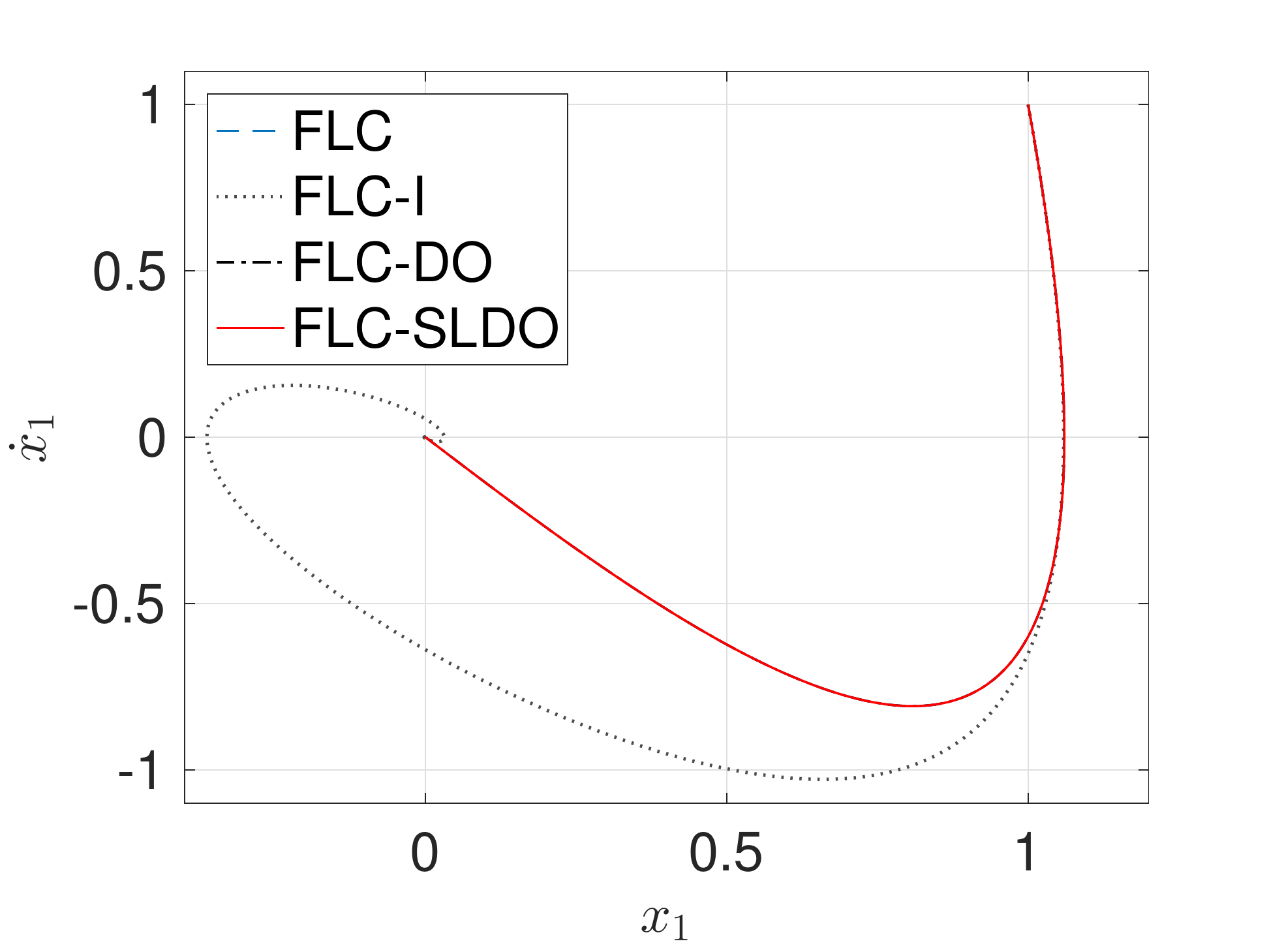}
\label{fig_x1dotx1_1}
}
\subfigure[ ]{
\includegraphics[width=0.66\columnwidth]{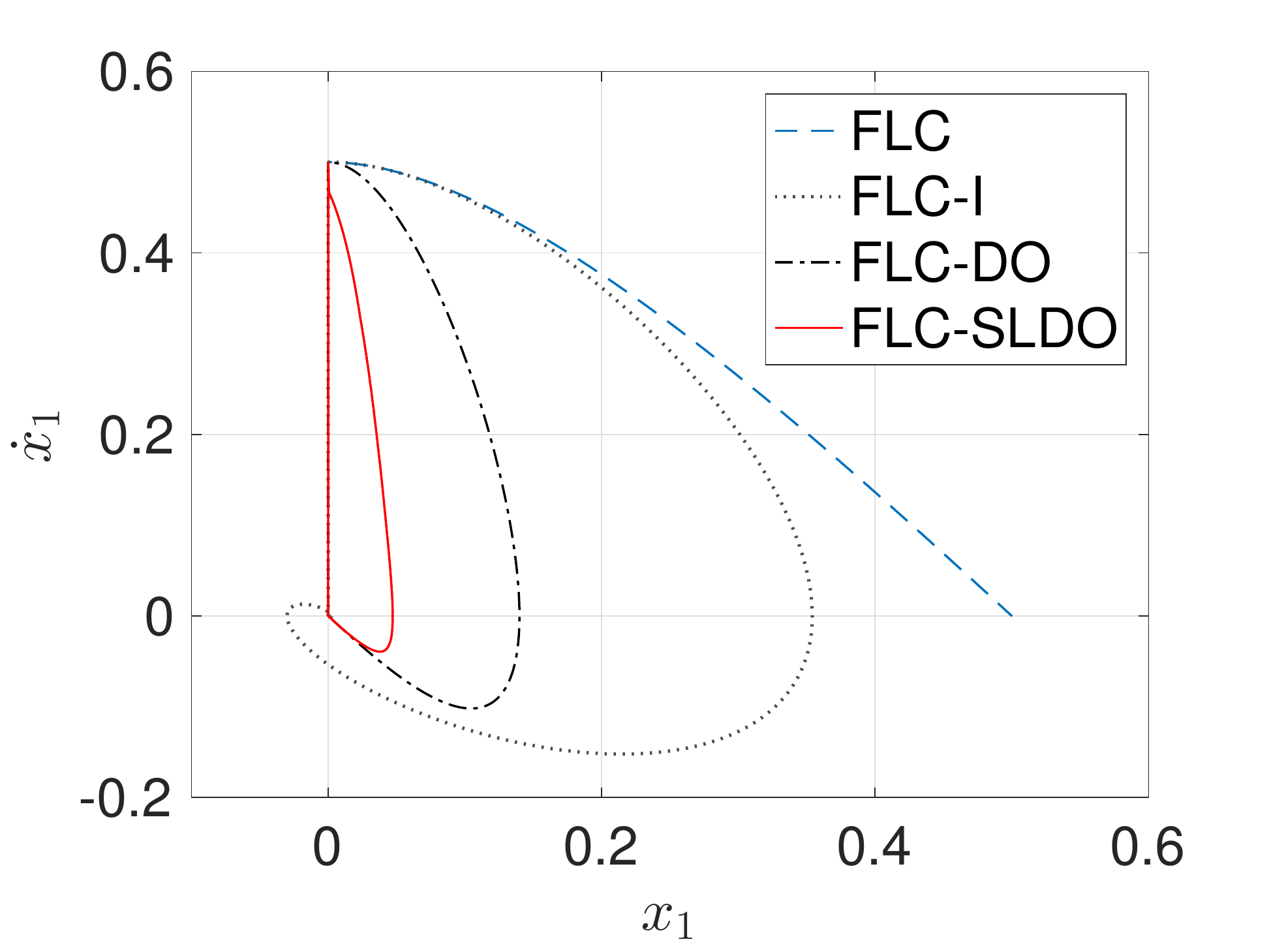}
\label{fig_x1dotx1_2}
}
\subfigure[ ]{
\includegraphics[width=0.66\columnwidth]{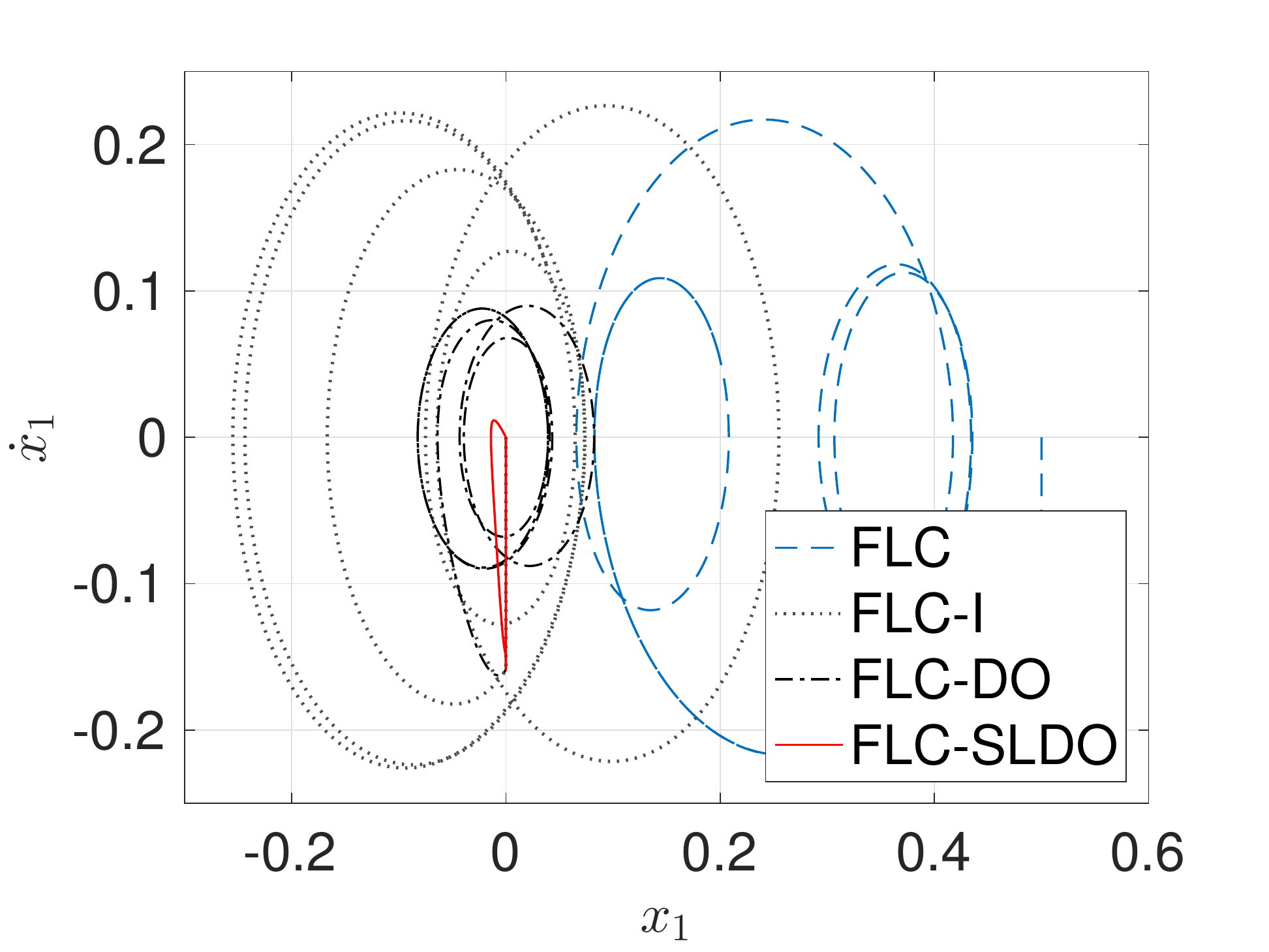}
\label{fig_x1dotx1_3}
}
\caption[Optional caption for list of figures]{ The phase portraits (a) $t=0-20$ second (b) $t=20-40$ second (c) $t=40-60$ second }
\label{fig_x1dotx1}
\end{figure*}

Type-2 fuzzy membership functions are used in the proposed SLDO structure. It is feasible to downgrade them to type-1 counterparts by equalizing the upper and lower values of parameters in \eqref{eq_c_1i}-\eqref{eq_sigma_2j_upper}. It is reported that since type-2 fuzzy logic systems have more degrees of freedom than type-1 counterparts, they have  the capability of dealing with noisy measurements and uncertainties more adequately \cite{Mendel2000913, Hsiao20081696, 5772027}. For this purpose, in order to compare the performance of T2NFS with its type-1 counterpart under noisy conditions, the outputs of the systems are measured with different noise levels $SNR$. The mean squared errors of disturbance estimation for the different noise levels are given in Table \ref{tab_mse}. As can be seen from this table, the T2NFS gives a smaller error than that of the type-1 neuro-fuzzy structure, and the performance of T2NFS is more remarkable while the noise level is high.
\begin{table}[h!]
\centering
\caption{Mean Squared Error.}\label{tab_mse}
\begin{tabular}{lccc}
  \hline
   & 20 (dB)& 40 (dB)  & 80 (dB)  \\
    \hline
      \hline
  Type-1 NFS & 0.6084 & 0.0133 & 0.0016 \\
  Type-2 NFS & 0.5542 & 0.0129 & 0.0016  \\
  \hline
    \hline
  Performance & \% 8.90  & \%  3.76 & \% 0 \\
      \hline
\end{tabular}
\end{table}

\section{CONCLUSION and FUTURE WORK}\label{sec_conc}

\subsection{Conclusion}
In this paper, the FLC-SLDO has been investigated to eliminate mismatched uncertainties on the system and to obtain robust control performance. The simulation results show that FLC-BNDO gives robust control performance against only mismatched time-invariant uncertainties. The stability of the training algorithm of the SLDO and the stability of the FLC-SLDO have been proven, and the simulation results show that FLC-SLDO gives robust control performance against mismatched time-varying uncertainties. Additionally, the FLC-SLDO gives less rise time, settling time and overshoot than the FLC-BNDO in the presence of mismatched time-invariant uncertainties, and they maintain the nominal control performance in the absence of mismatched uncertainties. Thanks to the T2NFS, the T2NFS working in parallel with the conventional estimation law in FEL scheme can estimate time-varying disturbances as distinct from the BNDO, and cope with uncertainties better than its type-1 counterpart under noisy conditions. The SMC-theory based learning algorithm requires significantly less computation time than the traditional ones, e.g., gradient descent and evolutionary training algorithms; therefore, the SLDO is computationally an efficient disturbance observer.

\subsection{Future Work}
The proposed method with its stability analysis is valid for
only a second-order nonlinear system with mismatched uncertainties.
As a future work, the extension of the SLDO and FLC law for nth order nonlinear systems with mismatched uncertainties are some interesting
topics to be investigated.

\section*{APPENDIX}
\subsection*{Calculation of $\dot{\tau}_{n}$}

By taking the time-derivative of the lower and upper Gaussian membership functions in \eqref{eq_mu1_lower}-\eqref{eq_mu2_upper}, the following equations are obtained as:
\begin{equation}
\dot{\underline{\mu}}_{1i}(x_1) = -2 \underline{N}_{1i} \dot{\underline{N}}_{1i} \underline{\mu}_{1i}(\xi_1),
\end{equation}
\begin{equation}
\dot{\overline{\mu}}_{1i}(x_1) = -2 \overline{N}_{1i} \dot{\overline{N}}_{1i} \overline{\mu}_{1i}(\xi_1)
\end{equation}
\begin{equation}
\dot{\underline{\mu}}_{2j}(x_2) = -2 \underline{N}_{2j} \dot{\underline{N}}_{2j} \underline{\mu}_{2j}(\xi_2)
\end{equation}
\begin{equation}
\dot{\overline{\mu}}_{2j}(x_2) = -2 \overline{N}_{2j} \dot{\overline{N}}_{2j} \overline{\mu}_{2j}(\xi_2)
\end{equation}
where
\begin{eqnarray}
&& \underline{N}_{1i}=\Big(\frac{\xi_1-\underline{c}_{1i}}{\underline{\sigma}_{1i}}\Big), \quad  \overline{N}_{1i}=\Big(\frac{\xi_1-\overline{c}_{1i}}{\overline{\sigma}_{1i}}\Big),  \\
 && \underline{N}_{2j}=\Big(\frac{\xi_2-\underline{c}_{2j}}{\underline{\sigma}_{2j}}\Big), \quad \overline{N}_{2j}=\Big(\frac{\xi_2-\overline{c}_{2j}}{\overline{\sigma}_{2j}}\Big)  \\ 
 && \dot{\underline{N}}_{1i}= \frac{(\dot{\xi_1} - \dot{\underline{c}}_{{1i}})\underline{\sigma}_{1i}-(\xi_1 - \underline{c}_{1i})\dot{\underline{\sigma}}_{1i}}{\underline{\sigma}^2 _{1i}}  \\
&& \dot{\overline{N}}_{1i}= \frac{(\dot{\xi_1} - \dot{\overline{c}}_{{1i}})\overline{\sigma}_{1i}-(\xi_1 - \overline{c}_{1i})\dot{\overline{\sigma}}_{1i}}{\overline{\sigma}^2 _{1i}} \\
&& \dot{\underline{N}}_{2i}= \frac{(\dot{\xi_2} - \dot{\underline{c}}_{{2i}})\underline{\sigma}_{2i}-(\xi_2 - \underline{c}_{2i})\dot{\underline{\sigma}}_{2i}}{\underline{\sigma}^2 _{2i}} \\
&& \dot{\overline{N}}_{2i}= \frac{(\dot{\xi_2} - \dot{\overline{c}}_{{2i}})\overline{\sigma}_{2i}-(\xi_2 - \overline{c}_{2i})\dot{\overline{\sigma}}_{2i}}{\overline{\sigma}^2 _{2i}}
\end{eqnarray}

By taking the time-derivative of the normalized firing strengths of the lower output signal in \eqref{eq_wij_lower_upper_normalized}, the following equation is obtained:

\begin{eqnarray}\label{dotwij_normalized}
\dot{\widetilde{\underline{w}}}_{ij} = -\widetilde{\underline{w}}_{ij}\underline{K}_{ij}+ \widetilde{\underline{w}}_{ij} \sum_{i=1}^{I}\sum_{j=1}^{J}\widetilde{\underline{w}}_{ij} \underline{K}_{ij} \nonumber \\
\dot{\widetilde{\overline{w}}}_{ij} = -\widetilde{\overline{w}}_{ij}\overline{K}_{ij} + \widetilde{\overline{w}}_{ij} \sum_{i=1}^{I}\sum_{j=1}^{J}\widetilde{\overline{w}}_{ij} \overline{K}_{ij}
\end{eqnarray}
where
\begin{eqnarray*}
\underline{K}_{ij} = 2 \Big(\underline{N}_{1i} \dot{\underline{N}}_{1i} + \underline{N}_{2j} \dot{\underline{N}}_{2j}\Big) = 4 \alpha sgn{(s)} \nonumber \\
\overline{K}_{ij} = 2 \Big(\overline{N}_{1i} \dot{\overline{N}}_{1i} +\overline{N}_{2j} \dot{\overline{N}}_{2j} \Big) = 4 \alpha sgn{(s)} 
\end{eqnarray*}
%

The time-derivative of \eqref{eq_taun} is obtained to find $\dot{\tau}_{n}$ as follows:
\begin{eqnarray}
\dot{\tau}_n & = & q\sum_{i=1}^{I}\sum_{j=1}^{J}(\dot{f}_{ij} \widetilde{\underline{w}}_{ij}+ f_{ij}\dot{\widetilde{\underline{w}}}_{ij}) \nonumber \\ 
&&
+(1-q)\sum_{i=1}^{I}\sum_{j=1}^{J}(\dot{f}_{ij}\widetilde{\overline{w}}_{ij} + f_{ij}\dot{\widetilde{\overline{w}}}_{ij}) \nonumber \\
& & +\dot{q} \sum_{i=1}^{I}\sum_{j=1}^{J}f_{ij}\widetilde{\underline{w}}_{ij}
-\dot{q} \sum_{i=1}^{I}\sum_{j=1}^{J}f_{ij}\widetilde{\overline{w}}_{ij}
\end{eqnarray}

If \eqref{dotwij_normalized} is inserted to the aforementioned equation:
\begin{eqnarray}\label{dotVc2}
\dot{\tau}_n & = & q\sum_{i=1}^{I}\sum_{j=1}^{J}\bigg(\Big(-\widetilde{\underline{w}}_{ij}\underline{K}_{ij}+\widetilde{\underline{w}}_{ij} \sum_{i=1}^{I}\sum_{j=1}^{J}\widetilde{\underline{w}}_{ij} \underline{K}_{ij} \Big)f_{ij}  \nonumber \\ 
&&+\widetilde{\underline{w}}_{ij}\dot{f}_{ij}\bigg) +(1-q)\sum_{i=1}^{I}\sum_{j=1}^{J}\bigg(\Big( -\widetilde{\overline{w}}_{ij} \overline{K}_{ij}  \nonumber \\ 
&& + \widetilde{\overline{w}}_{ij} \sum_{i=1}^{I}\sum_{j=1}^{J}\widetilde{\overline{w}}_{ij}  \overline{K}_{ij} \Big)f_{ij} +\widetilde{\overline{w}}_{ij}\dot{f}_{ij}\bigg)  \nonumber \\ 
&&+\dot{q} \sum_{i=1}^{I}\sum_{j=1}^{J}f_{ij} (\widetilde{\underline{w}}_{ij} - \widetilde{\overline{w}}_{ij} )   \nonumber \\
& = & q \sum_{i=1}^{I}\sum_{j=1}^{J} 4 \alpha sgn{(s)}  \bigg(\Big(-\widetilde{\underline{w}}_{ij} + \widetilde{\underline{w}}_{ij} \sum_{i=1}^{I}\sum_{j=1}^{J}\widetilde{\underline{w}}_{ij} \Big)f_{ij} \nonumber \\
&& +\widetilde{\underline{w}}_{ij}\dot{f}_{ij}\bigg) +(1-q)\sum_{i=1}^{I}\sum_{j=1}^{J} 4 \alpha sgn{(s)} \bigg(\Big( -\widetilde{\overline{w}}_{ij} \nonumber \\
&& + \widetilde{\overline{w}}_{ij} \sum_{i=1}^{I}\sum_{j=1}^{J}\widetilde{\overline{w}}_{ij}  \Big)f_{ij} +\widetilde{\overline{w}}_{ij}\dot{f}_{ij}\bigg) \nonumber \\
&& +\dot{q} \sum_{i=1}^{I}\sum_{j=1}^{J}f_{ij} (\widetilde{\underline{w}}_{ij} - \widetilde{\overline{w}}_{ij} ) 
\end{eqnarray}

Since $\sum_{i=1}^{I}\sum_{j=1}^{J}\widetilde{\overline{w}}_{ij} =1$ and $\sum_{i=1}^{I}\sum_{j=1}^{J}\widetilde{\underline{w}}_{ij} =1$, the aforementioned equation becomes by using \eqref{eq_f_ij} and \eqref{eq_q} as follows:
\begin{eqnarray}\label{dotVc4}
\dot{\tau}_n & = & \sum_{i=1}^{I}\sum_{j=1}^{J}\Big(q\widetilde{\underline{w}}_{ij} \dot{f}_{ij} + (1-q)\widetilde{\overline{w}}_{ij} \dot{f}_{ij} \Big) \nonumber \\
&& +\dot{q} \sum_{i=1}^{I}\sum_{j=1}^{J}f_{ij} (\widetilde{\underline{w}}_{ij} - \widetilde{\overline{w}}_{ij} )   \nonumber \\
& = & \bigg( q \sum_{i=1}^{I}\sum_{j=1}^{J}\widetilde{\underline{w}}_{ij}+ (1-q)\sum_{i=1}^{I}\sum_{j=1}^{J}\widetilde{\overline{w}}_{ij} \bigg) \dot{f}_{ij}  \nonumber \\
&& +\dot{q} \sum_{i=1}^{I}\sum_{j=1}^{J}f_{ij} (\widetilde{\underline{w}}_{ij} - \widetilde{\overline{w}}_{ij} )   \nonumber \\
& = &- 2 \alpha sgn{(s)}
\end{eqnarray}

\bibliography{reference}

\begin{thebibliography}{10}
\providecommand{\url}[1]{#1}
\csname url@samestyle\endcsname
\providecommand{\newblock}{\relax}
\providecommand{\bibinfo}[2]{#2}
\providecommand{\BIBentrySTDinterwordspacing}{\spaceskip=0pt\relax}
\providecommand{\BIBentryALTinterwordstretchfactor}{4}
\providecommand{\BIBentryALTinterwordspacing}{\spaceskip=\fontdimen2\font plus
\BIBentryALTinterwordstretchfactor\fontdimen3\font minus
  \fontdimen4\font\relax}
\providecommand{\BIBforeignlanguage}[2]{{%
\expandafter\ifx\csname l@#1\endcsname\relax
\typeout{** WARNING: IEEEtran.bst: No hyphenation pattern has been}%
\typeout{** loaded for the language `#1'. Using the pattern for}%
\typeout{** the default language instead.}%
\else
\language=\csname l@#1\endcsname
\fi
#2}}
\providecommand{\BIBdecl}{\relax}
\BIBdecl

\bibitem{7302059}
E.~Kayacan, H.~Ramon, and W.~Saeys, ``Robust trajectory tracking error
  model-based predictive control for unmanned ground vehicles,''
  \emph{IEEE/ASME Transactions on Mechatronics}, vol.~21, no.~2, pp. 806--814,
  April 2016.

\bibitem{ASJC1649}
J.~Y. Lau, W.~Liang, H.~C. Liaw, and K.~K. Tan, ``Sliding mode disturbance
  observer-based motion control for a piezoelectric actuator-based surgical
  device,'' \emph{Asian Journal of Control}, asjc.1649.

\bibitem{KAYACAN20141}
E.~Kayacan, E.~Kayacan, H.~Ramon, and W.~Saeys, ``Nonlinear modeling and
  identification of an autonomous tractor-trailer system,'' \emph{Computers and
  Electronics in Agriculture}, vol. 106, pp. 1 -- 10, 2014.

\bibitem{ASJC1489}
T.~Sun, J.~Zhang, and Y.~Pan, ``Active disturbance rejection control of surface
  vessels using composite error updated extended state observer,'' \emph{Asian
  Journal of Control}, vol.~19, no.~5, pp. 1802--1811, 2017.

\bibitem{7934317}
E.~Kayacan, ``Multiobjective $h_{\infty }$ control for string stability of
  cooperative adaptive cruise control systems,'' \emph{IEEE Transactions on
  Intelligent Vehicles}, vol.~2, no.~1, pp. 52--61, March 2017.

\bibitem{ASJC1347}
M.-X. Cheng and X.-H. Jiao, ``Observer-based adaptive l2 disturbance
  attenuation control of semi-active suspension with mr damper,'' \emph{Asian
  Journal of Control}, vol.~19, no.~1, pp. 346--355, 2017.

\bibitem{Gao2017}
F.~Gao, M.~Wu, J.~She, and W.~Cao, ``Active disturbance rejection in affine
  nonlinear systems based on equivalent-input-disturbance approach,''
  \emph{Asian Journal of Control}, vol.~19, no.~5, pp. 1767--1776, 2017.

\bibitem{KAYACAN2017}
E.~Kayacan, J.~M. Peschel, and G.~Chowdhary, ``A self-learning disturbance
  observer for nonlinear systems in feedback-error learning scheme,''
  \emph{Engineering Applications of Artificial Intelligence}, vol.~62, pp. 276
  -- 285, 2017.

\bibitem{hall2006sliding}
C.~E. Hall and Y.~B. Shtessel, ``Sliding mode disturbance observer-based
  control for a reusable launch vehicle,'' \emph{Journal of guidance, control,
  and dynamics}, vol.~29, no.~6, pp. 1315--1328, 2006.

\bibitem{Komada1991}
S.~Komada, M.~Ishida, K.~Ohnishi, and T.~Hori, ``Disturbance observer-based
  motion control of direct drive motors,'' \emph{IEEE Transactions on Energy
  Conversion}, vol.~6, no.~3, pp. 553--559, Sep 1991.

\bibitem{liu2000disturbance}
C.-S. Liu and H.~Peng, ``Disturbance observer based tracking control,''
  \emph{Journal of Dynamic Systems, Measurement, and Control}, vol. 122, no.~2,
  pp. 332--335, 2000.

\bibitem{Yonghwan1999}
Y.~Oh and W.~K. Chung, ``Disturbance-observer-based motion control of redundant
  manipulators using inertially decoupled dynamics,'' \emph{IEEE/ASME
  Transactions on Mechatronics}, vol.~4, no.~2, pp. 133--146, Jun 1999.

\bibitem{Chen20091205}
X.~Chen, J.~Yang, S.~Li, and Q.~Li, ``Disturbance observer based multi-variable
  control of ball mill grinding circuits,'' \emph{Journal of Process Control},
  vol.~19, no.~7, pp. 1205 -- 1213, 2009.

\bibitem{Yokoyama1994}
T.~Yokoyama and A.~Kawamura, ``Disturbance observer based fully digital
  controlled pwm inverter for cvcf operation,'' \emph{IEEE Transactions on
  Power Electronics}, vol.~9, no.~5, pp. 473--480, Sep 1994.

\bibitem{Chen2010}
M.~Chen and W.-H. Chen, ``Sliding mode control for a class of uncertain
  nonlinear system based on disturbance observer,'' \emph{International Journal
  of Adaptive Control and Signal Processing}, vol.~24, no.~1, pp. 51--64, 2010.

\bibitem{He2015}
W.~He and S.~S. Ge, ``Vibration control of a nonuniform wind turbine tower via
  disturbance observer,'' \emph{IEEE/ASME Transactions on Mechatronics},
  vol.~20, no.~1, pp. 237--244, Feb 2015.

\bibitem{Mattavelli2005}
P.~Mattavelli, ``An improved deadbeat control for ups using disturbance
  observers,'' \emph{IEEE Transactions on Industrial Electronics}, vol.~52,
  no.~1, pp. 206--212, Feb 2005.

\bibitem{Du2016}
J.~Du, X.~Hu, M.~Krstić, and Y.~Sun, ``Robust dynamic positioning of ships
  with disturbances under input saturation,'' \emph{Automatica}, vol.~73, pp.
  207 -- 214, 2016.

\bibitem{Leu1999}
Y.-G. Leu, T.-T. Lee, and W.-Y. Wang, ``Observer-based adaptive fuzzy-neural
  control for unknown nonlinear dynamical systems,'' \emph{IEEE Transactions on
  Systems, Man, and Cybernetics, Part B (Cybernetics)}, vol.~29, no.~5, pp.
  583--591, Oct 1999.

\bibitem{Yang2013}
J.~Yang, S.~Li, and X.~Yu, ``Sliding-mode control for systems with mismatched
  uncertainties via a disturbance observer,'' \emph{Industrial Electronics,
  IEEE Transactions on}, vol.~60, no.~1, pp. 160--169, Jan 2013.

\bibitem{KAYACAN201578}
E.~Kayacan, E.~Kayacan, H.~Ramon, and W.~Saeys, ``Towards agrobots:
  Identification of the yaw dynamics and trajectory tracking of an autonomous
  tractor,'' \emph{Computers and Electronics in Agriculture}, vol. 115, pp. 78
  -- 87, 2015.

\bibitem{ginoya2015state}
D.~Ginoya, P.~Shendge, and S.~Phadke, ``State and extended disturbance observer
  for sliding mode control of mismatched uncertain systems,'' \emph{Journal of
  Dynamic Systems, Measurement, and Control}, vol. 137, no.~7, p. 074502, 2015.

\bibitem{Yang2017}
J.~Yang, C.~Zhang, S.~Li, and X.~Niu, ``Semi-global exquisite disturbance
  attenuation control for perturbed uncertain nonlinear systems,'' \emph{Asian
  Journal of Control}, vol.~19, no.~4, pp. 1608--1619, 2017.

\bibitem{Ginoya2014}
D.~Ginoya, P.~Shendge, and S.~Phadke, ``Sliding mode control for mismatched
  uncertain systems using an extended disturbance observer,'' \emph{Industrial
  Electronics, IEEE Transactions on}, vol.~61, no.~4, pp. 1983--1992, April
  2014.

\bibitem{Chen2016}
C.~C. Chen, S.~S.~D. Xu, and Y.~W. Liang, ``Study of nonlinear integral sliding
  mode fault-tolerant control,'' \emph{IEEE/ASME Transactions on Mechatronics},
  vol.~21, no.~2, pp. 1160--1168, April 2016.

\bibitem{7498627}
J.~Zhang, X.~Liu, Y.~Xia, Z.~Zuo, and Y.~Wang, ``Disturbance observer-based
  integral sliding-mode control for systems with mismatched disturbances,''
  \emph{IEEE Transactions on Industrial Electronics}, vol.~63, no.~11, pp.
  7040--7048, Nov 2016.

\bibitem{SUN20141027}
H.~Sun and L.~Guo, ``Composite adaptive disturbance observer based control and
  back-stepping method for nonlinear system with multiple mismatched
  disturbances,'' \emph{Journal of the Franklin Institute}, vol. 351, no.~2,
  pp. 1027 -- 1041, 2014.

\bibitem{Wei2010}
X.~Wei and L.~Guo, ``Composite disturbance-observer-based control and h∞
  control for complex continuous models,'' \emph{International Journal of
  Robust and Nonlinear Control}, vol.~20, no.~1, pp. 106--118, 2010.

\bibitem{7265050}
W.~H. Chen, J.~Yang, L.~Guo, and S.~Li, ``Disturbance-observer-based control
  and related methods - an overview,'' \emph{IEEE Transactions on Industrial
  Electronics}, vol.~63, no.~2, pp. 1083--1095, Feb 2016.

\bibitem{erkancenmpc}
E.~Kayacan, E.~Kayacan, H.~Ramon, and W.~Saeys, ``Learning in centralized
  nonlinear model predictive control: Application to an autonomous
  tractor-trailer system,'' \emph{IEEE Transactions on Control Systems
  Technology}, vol.~23, no.~1, pp. 197--205, Jan 2015.

\bibitem{erkan2016acc}
E.~Kayacan, J.~M. Peschel, and E.~Kayacan, ``Centralized, decentralized and
  distributed nonlinear model predictive control of a tractor-trailer system: A
  comparative study,'' in \emph{2016 American Control Conference (ACC)}, July
  2016, pp. 4403--4408.

\bibitem{erkandenmpc}
E.~Kayacan, E.~Kayacan, H.~Ramon, and W.~Saeys, ``Robust tube-based
  decentralized nonlinear model predictive control of an autonomous
  tractor-trailer system,'' \emph{IEEE/ASME Transactions on Mechatronics},
  vol.~20, no.~1, pp. 447--456, Feb 2015.

\bibitem{erkandinmpc}
------, ``Distributed nonlinear model predictive control of an autonomous
  tractor-trailer system,'' \emph{Mechatronics}, vol.~24, no.~8, pp. 926 --
  933, 2014.

\bibitem{Haseltine2005}
E.~L. Haseltine, , and J.~B. Rawlings*, ``Critical evaluation of extended
  kalman filtering and moving-horizon estimation,'' \emph{Industrial \&
  Engineering Chemistry Research}, vol.~44, no.~8, pp. 2451--2460, 2005.

\bibitem{Rawlings2006}
J.~B. Rawlings and B.~R. Bakshi, ``Particle filtering and moving horizon
  estimation,'' \emph{Computers \& Chemical Engineering}, vol.~30, no. 10–12,
  pp. 1529 -- 1541, 2006.

\bibitem{Efe2000}
M.~O. Efe, O.~Kaynak, and X.~Yu, ``Sliding mode control of a three degrees of
  freedom anthropoid robot by driving the controller parameters to an
  equivalent regime,'' \emph{Journal of Dynamic Systems, Measurement, and
  Control}, vol. 122, no.~4, pp. 632--640, 2000.

\bibitem{erkan2013t1nfc}
E.~Kayacan, E.~Kayacan, H.~Ramon, and W.~Saeys, ``Adaptive neuro-fuzzy control
  of a spherical rolling robot using sliding-mode-control-theory-based online
  learning algorithm,'' \emph{Cybernetics, IEEE Transactions on}, vol.~43,
  no.~1, pp. 170--179, Feb 2013.

\bibitem{erdal2015t2fnn}
E.~Kayacan, E.~Kayacan, H.~Ramon, O.~Kaynak, and W.~Saeys, ``Towards agrobots:
  Trajectory control of an autonomous tractor using type-2 fuzzy logic
  controllers,'' \emph{Mechatronics, IEEE/ASME Transactions on}, vol.~20,
  no.~1, pp. 287--298, Feb 2015.

\bibitem{6250790}
E.~Kayacan, W.~Saeys, E.~Kayacan, H.~Ramon, and O.~Kaynak, ``Intelligent
  control of a tractor-implement system using type-2 fuzzy neural networks,''
  in \emph{2012 IEEE International Conference on Fuzzy Systems}, June 2012, pp.
  1--8.

\bibitem{6871316}
E.~Kayacan, E.~Kayacan, and M.~A. Khanesar, ``Identification of nonlinear
  dynamic systems using type-2 fuzzy neural networks - a novel learning
  algorithm and a comparative study,'' \emph{IEEE Transactions on Industrial
  Electronics}, vol.~62, no.~3, pp. 1716--1724, March 2015.

\bibitem{Gomi1993}
H.~Gomi and M.~Kawato, ``Neural network control for a closed-loop system using
  feedback-error-learning,'' \emph{Neural Networks}, vol.~6, no.~7, pp. 933 --
  946, 1993.

\bibitem{chen2003nonlinear}
W.-H. Chen, ``Nonlinear disturbance observer-enhanced dynamic inversion control
  of missiles,'' \emph{Journal of Guidance, Control, and Dynamics}, vol.~26,
  no.~1, pp. 161--166, 2003.

\bibitem{Chen2004}
------, ``Disturbance observer based control for nonlinear systems,''
  \emph{Mechatronics, IEEE/ASME Transactions on}, vol.~9, no.~4, pp. 706--710,
  Dec 2004.

\bibitem{Yang2011}
J.~Yang, W.~H. Chen, and S.~Li, ``Non-linear disturbance observer-based robust
  control for systems with mismatched disturbances/uncertainties,'' \emph{IET
  Control Theory Applications}, vol.~5, no.~18, pp. 2053--2062, December 2011.

\bibitem{khalil1996nonlinear}
H.~K. Khalil and J.~Grizzle, \emph{Nonlinear systems}.\hskip 1em plus 0.5em
  minus 0.4em\relax Prentice hall New Jersey, 1996, vol.~3.

\bibitem{7902222}
X.~Yu and O.~Kaynak, ``Sliding mode control made smarter: A computational
  intelligence perspective,'' \emph{IEEE Systems, Man, and Cybernetics
  Magazine}, vol.~3, no.~2, pp. 31--34, April 2017.

\bibitem{Mendel2000913}
J.~M. Mendel, ``Uncertainty, fuzzy logic, and signal processing,'' \emph{Signal
  Processing}, vol.~80, no.~6, pp. 913 -- 933, 2000.

\bibitem{Hsiao20081696}
M.-Y. Hsiao, T.-H.~S. Li, J.-Z. Lee, C.-H. Chao, and S.-H. Tsai, ``Design of
  interval type-2 fuzzy sliding-mode controller,'' \emph{Information Sciences},
  vol. 178, no.~6, pp. 1696 -- 1716, 2008.

\bibitem{5772027}
M.~Khanesar, E.~Kayacan, M.~Teshnehlab, and O.~Kaynak, ``Analysis of the noise
  reduction property of type-2 fuzzy logic systems using a novel type-2
  membership function,'' \emph{IEEE Transactions on Systems, Man, and
  Cybernetics, Part B: Cybernetics}, vol.~41, no.~5, pp. 1395--1406, Oct 2011.

\end{thebibliography}
\bibliographystyle{IEEEtran}

\end{document}